\documentclass[a4paper]{article}
\usepackage{a4wide}
\usepackage[UKenglish]{babel}
\usepackage{graphicx,subfigure,caption}
\graphicspath{{./}{figs/}}
\usepackage[colorlinks=true,linkcolor=blue,citecolor=red]{hyperref}
\usepackage{xcolor}
\usepackage{amssymb,amsthm,empheq,bbold}
\usepackage[inline]{enumitem}
\usepackage{authblk}
\usepackage{calrsfs}
    \DeclareMathAlphabet{\pazocal}{OMS}{zplm}{m}{n}

\DeclarePairedDelimiter\abs{\lvert}{\rvert}
\DeclarePairedDelimiter\ave{\langle}{\rangle}
\newcommand{\bA}{\mathbf{A}}
\newcommand{\cd}{\mathfrak{d}}
\newcommand{\dWinf}{\underline\cd_W}
\newcommand{\dWNinf}{\underline\cd_{W_N}}
\newcommand{\cE}{\pazocal{E}}
\newcommand{\cG}{\pazocal{G}}
\newcommand{\Leb}{\pazocal{L}^1}
\newcommand{\Lip}{\operatorname{Lip}}
\newcommand{\cm}{\mathfrak{m}}
\newcommand{\cM}{\mathfrak{M}}
\newcommand{\N}{\mathbb{N}}
\newcommand{\norm}[2]{\Vert#1\Vert_{#2}}
\newcommand{\cP}{\mathcal{P}}
\newcommand{\pr}[1]{{}^\prime\!#1}
\newcommand{\R}{\mathbb{R}}
\newcommand{\supp}[1]{\operatorname{supp}#1}
\newcommand{\cV}{\pazocal{V}}
\newcommand{\cW}{\pazocal{W}}

\newtheorem{assumption}{Assumption}[section]

\newtheorem{theorem}[assumption]{Theorem}
\newtheorem{lemma}[assumption]{Lemma}

\theoremstyle{remark}\newtheorem{remark}[assumption]{Remark}
\allowdisplaybreaks

\title{Homogeneous Boltzmann-type equations on dense graphs}
\author{Gabriele Taricco}
\author{Andrea Tosin}
\affil{{\small Department of Mathematical Sciences ``G. L. Lagrange'', Politecnico di Torino, Italy}}
\date{}
		
\begin{document}
\maketitle
	
\begin{abstract}
In kinetic theory, interactions between particles are typically assumed to be ``all-to-all'', meaning that any pair of randomly selected particles may, in principle, interact. This assumption originates from the theory of colliding gas molecules; however, it may be less appropriate for describing other forms of interaction, such as social interactions. These are more naturally characterised as ``some-to-some'', reflecting the existence of preferential connections between agents. In this paper, we consider homogeneous Boltzmann-type equations on finite graphs that model such networks of preferential interactions, and we rigorously derive their dense graph limit as the number of agents tends to infinity. We also investigate the long-time behaviour of the limiting equation in the case of linear pairwise interactions, characterising the emergent equilibrium distributions and relating them to their counterparts in the classical ``all-to-all'' setting.

\medskip


\medskip

\end{abstract}

\section{Introduction}
In this paper, we study homogeneous Boltzmann-type equations on graphs as a model for networked interactions in multi-agent systems.

In the classical kinetic theory of gases, any pair of freely moving molecules may collide. Over the past decades, this paradigm has been extended to the kinetic modelling of collective behaviour in social systems, tacitly assuming that interactions are ``all-to-all'', that is that they occur, in principle, between any randomly selected pair of agents. See e.g.,~\cite{pareschi2013BOOK}. Such an assumption, however, overlooks the fact that social interactions are typically ``some-to-some'', namely constrained by an underlying network of preferential connections, which determines the set of feasible communications among the agents. In particular, agents can effectively exchange their characteristic traits only if they are directly connected or, in the terminology of social networks, if they ``follow'' each other.

This naturally raises the question of how the topology of the interaction network can be incorporated into a homogeneous Boltzmann-type framework, namely a statistical description of the evolution of the distribution function of the agents' characteristic trait under pairwise interactions. It is worth emphasising that the term \textit{homogeneous} refers here to the absence of the spatial transport term that appears in the Boltzmann equation of gas dynamics. The introduction of an interaction network does not, in itself, reintroduce spatial dynamics, as the network encodes latent connection patterns rather than spatial dynamics. Consequently, the Boltzmann-type equations on graphs considered in this paper remain homogeneous with respect to space.

In the literature, network interactions have typically been incorporated in kinetic models of sociophysical systems by treating agents' connectivity as a structural variable coupled with their characteristic trait and endowed with a probability density function. This density may either be prescribed \textit{a priori}, drawing inspiration from models describing the scaling properties of large-scale networks, such as~\cite{barabasi1999SCIENCE,barabasi1999PHYSA,franceschi2025PRSA}, or arise as the solution of a possibly coupled equation governing the evolution of social contacts; see, for instance,~\cite{albi2025EJAM,burger2025SIADS,he2023JAMP,loy2022PTRSA,toscani2018PRE}.

More recently, \textit{graphon}-based approaches have been proposed to account for interaction networks in Boltzmann-type equations. Roughly speaking, a graphon may be regarded as a continuum counterpart of the adjacency matrix of a graph in the limit of a large number of vertices and edges, when the graph becomes \textit{dense} and can therefore be represented by a continuous object encoding the density of connections. The notion of graphon was introduced by Lov\'{a}sz and Szegedy in their seminal work on limits of dense graph sequences~\cite{lovasz2006JCTB}. To date, graphons have mainly been employed as phenomenological interaction kernels in Boltzmann-type models of opinion formation~\cite{duering2024JNS}, possibly coupled with the spread of infectious diseases~\cite{ali2026PREPRINT,bondesan2026PREPRINT}. See also~\cite{burger2021VJM}, where a closely related formalism is adopted, albeit without any explicit reference to graphons; and~\cite{prisant2026TCNS}, where the evolution of the mean opinion on large undirected networks is studied through a framework bearing some similarities to an averaged Boltzmann-type kinetic description.

The present paper contributes to this emerging line of research in two complementary ways.

First, in Section~\ref{sect:kin_eq.graph}, we provide a rigorous derivation of the homogeneous Boltzmann-type equation on dense graphs as the limit of corresponding kinetic equations posed on a sequence of finite graphs. In the latter setting, the network of interactions is described by adjacency matrices, making these models more fundamental and closer to the physical interpretation of networked interactions. To this end, we build upon the framework introduced in~\cite{nurisso2024EJAM}, complementing it with concepts from the theory of graph limits. Standard references for the latter include the aforementioned paper by Lov\'{a}sz and Szegedy~\cite{lovasz2006JCTB} and Lov\'{a}sz's book~\cite{lovasz2012BOOK}. We also draw upon the purposeful monograph~\cite{janson2013NYJM}.

Secondly, in Section~\ref{sect:relax_equil}, focusing on the limiting equation in the benchmark case of linear interaction rules, we investigate some classes of stationary distributions arising for long times. These play a role analogous to that of Maxwellian distributions in classical kinetic theory. Their analysis allows us to clarify the relationship between solutions on dense graphs and solutions of the classical homogeneous Boltzmann-type equation without graph structure. In particular, we prove that, for suitable regimes of the interaction parameters, graphon-based solutions converge asymptotically in time to their graph-free counterparts.

\section{Formulation of kinetic equations on graphs}
\label{sect:kin_eq.graph}
In this section, we rigorously derive a homogeneous Boltzmann-type equation on a \textit{dense} graph as the limit of analogous equations posed on a sequence $\{\cG_N\}_{N\in\N}$ of \textit{finite} graphs with an increasing number $N$ of vertices. Each graph of the sequence is identified by the set of vertices $\cV_N=\{1,\,\dots,\,N\}$ and by a set of edges $\cE_N\subseteq\cV_N^2$. This information is encoded in the adjacency matrix $\bA_N\in\R^{N\times N}$ of $\cG_N$, whose generic entry $a^N_{ij}$ is
$$ a^N_{ij}=
    \begin{cases}
        1 & \text{if } (i,\,j)\in\cE_N \\
        0 & \text{otherwise.}
    \end{cases} $$
For the sake of simplicity, we consider \textit{undirected} graphs, which are such that if an edge exists between any two vertices $i$, $j$, i.e. $(i,\,j)\in\cE_N$, then it also serves as an edge between the vertices $j$, $i$, i.e. $(j,\,i)\in\cE_N$. As a consequence, the adjacency matrix $\bA_N$ is symmetric.

Recent literature provides applications of Boltzmann-type kinetic equations on dense graphs to problems in opinion formation~\cite{duering2024JNS} and epidemiology~\cite{ali2026PREPRINT,bondesan2026PREPRINT}. In those works, such equations are introduced heuristically and used as building blocks for more elaborated models. Here, by contrast, we focus on their rigorous justification as the limit, as $N\to\infty$, of less abstract kinetic descriptions consisting of finite systems of classical homogeneous Boltzmann-type equations coupled through the graph structure. This also reveals interesting applications of analytical techniques from graph theory to kinetic equations.

\subsection{An overview of graphon theory}
To each adjacency matrix $\bA_N\in\R^{N\times N}$, we associate a two-variable function $W_N:[0,\,1]^2\to\{0,\,1\}$, constructed as follows. First, we consider the partition of the interval $[0,\,1]\subset\R$ made by the subintervals
\begin{equation}
    I_1:=\left[0,\,\frac{1}{N}\right], \qquad I_i:=\left(\frac{i-1}{N},\,\frac{i}{N}\right], \quad i=2,\,\dots,\,N;
    \label{eq:Ii}
\end{equation}
then we set
\begin{equation}
    W_N(x,x_\ast):=\sum_{i=1}^{N}\sum_{j=1}^{N}a^N_{ij}\chi_{I_i}(x)\chi_{I_j}(x_\ast).
    \label{eq:WN}
\end{equation}
Consequently, $W_N$ is a piecewise constant function on the squares $I_i\times I_j\subset [0,\,1]^2$ taking in each of them the values $a^N_{ij}$: 
$$ W_N(x,x_\ast)=a^N_{ij} \quad \text{for} \quad (x,\,x_\ast)\in I_i\times I_j, \quad i,\,j=1,\,\dots,\,N. $$
In practice, $W_N$ maps $\bA_N$ on a pixelation of the unit square in $\R^2$. The symmetry of $\bA_N$ implies that also $W_N$ is a symmetric function, i.e. $W_N(x,x_\ast)=W_N(x_\ast,x)$ for all $x,\,x_\ast\in [0,\,1]$.

As the number $N$ of vertices of the graph grows, such pixelation gets finer and finer. The question then arises whether, for $N\to\infty$, the $W_N$'s converge to any function $W:[0,\,1]^2\to [0,\,1]$ representing the connectivity of a \textit{dense} limit graph for which neither the vertices nor the edges are described pointwise. If this is the case, the function $W$ is called a \textit{graphon}, a term suggestive of a contraction of ``\textit{graph} functi\textit{on}''.

To answer this question, an appropriate concept of convergence needs to be introduced. In graph theory, one of the most popular of such concepts is the one based on the so-called \textit{cut norm}, denoted $\norm{\cdot}{\Box}$ and defined, for every function $W\in L^\infty([0,\,1]^2)$, as
\begin{equation}
    \norm{W}{\Box}:=\sup_{U,\,V\subseteq [0,\,1]}\abs*{\iint_{U\times V}W(x,x_\ast)\,dx\,dx_\ast}.
    \label{eq:cut_norm}
\end{equation}
Other definitions are possible, see~\cite{janson2013NYJM} for details, but the one above has the merit of becoming particularly expressive when applied to $W_N$:
$$ \norm{W_N}{\Box}=\sup_{U,\,V\subseteq [0,\,1]}\abs*{\sum_{i=1}^{N}\sum_{j=1}^{N}a^N_{ij}\Leb(U\cap I_i)\Leb(V\cap I_j)}, $$
where $\Leb$ denotes the Lebesgue measure on $\R$. From here, we see that, informally speaking, the cut norm of $W_N$ counts the number of edges between pairs of vertices of $\cG_N$ in every rectangular subregion of the graph to detect the largest \textit{density of connections}.

The convergence of the sequence $\{W_N\}_{N\in\N}$ to a graphon $W$ in the cut norm means, as usual, that $\norm{W-W_N}{\Box}\to 0$ when $N\to\infty$. In particular, from~\eqref{eq:cut_norm} it is immediate to establish that $\norm{W-W_N}{\Box}\leq\norm{W-W_N}{L^1([0,\,1]^2)}$, hence that the convergence in the cut norm is weaker than the convergence in the classical $L^1$-norm on the unit square. A less straightforward property, which holds thanks to the piecewise constant form of each $W_N$ on a uniform pixelation of the unit square, is that
\begin{equation}
    \norm{W-W_N}{L^1([0,\,1]^2)}\leq\sqrt{2N}\norm{W-W_N}{\Box},
    \label{eq:L1<=cut}
\end{equation}
see~\cite[Lemma~10.7 and Remark~10.8]{janson2013NYJM}, whence we deduce that the convergence of the $W_N$'s in the cut norm may imply the convergence in the $L^1$-norm provided it is sufficiently fast.

Interestingly, the limit graphon $W$ may not take only the values $0$ and $1$ like the $W_N$'s. Its image may be, in general, a Lebesgue-non-negligible subset of the interval $[0,\,1]$, the values $W(x,x_\ast)\in (0,\,1)$ being interpreted as the probability of an edge between the points $x,\,x_\ast\in [0,\,1]$ in the dense limit graph. Applying such an interpretation back to $W_N$, one usually calls the latter a \textit{random-free $N$-step graphon}, where ``random-free'' refers to the fact that $W_N(x,x_\ast)\in\{0,\,1\}$, hence that the probability of an edge between the points $x,\,x_\ast$ is either $0$ or $1$.

The variable $x\in [0,\,1]$ is usually referred to as a \textit{latent variable}, which parametrises the pattern of connections within the communities described by the graphons $W_N$ and their limit $W$. This variable can often be endowed with a meaningful modelling interpretation in specific applications. We will return to this point in Section~\ref{sect:ex_graphons} below.

A useful quantity that can be computed from a graphon $W$ is the \textit{degree function}
$$ \cd_W(x):=\int_0^1W(x,x_\ast)\,dx_\ast, $$
which represents the \textit{mean degree} of the point $x$ in the dense limit graph. To better understand this interpretation, it is instructive to compute such quantity for $W_N$. We have:
$$ \cd_{W_N}(x)=\int_0^1W_N(x,x_\ast)\,dx_\ast=\frac{1}{N}\sum_{i=1}^{N}\sum_{j=1}^{N}a^N_{ij}\chi_{I_i}(x), $$
whence, if $\bar{\imath}\in\{1,\,\dots,\,N\}$ is such that $x\in I_{\bar{\imath}}$, we get
$$ \cd_{W_N}(x)=\frac{1}{N}\sum_{j=1}^{N}a^N_{\bar{\imath}j}=\frac{1}{N}\deg{(\bar{\imath})}. $$

Observe that $0\leq\cd_W(x)\leq 1$ for every $x\in [0,\,1]$. Moreover, associated with the degree function, we introduce the following \textit{connectivity parameters} of the graphon:
$$ \dWinf:=\inf_{x\in [0,\,1]}\cd_W(x)\in [0,\,1], \qquad \norm{\cd_W}{\infty}:=\sup_{x\in [0,\,1]}\cd_W(x)\in [0,\,1], $$
which will play a key role in the subsequent analysis.

The construction of the $W_N$'s is intrinsically tied to the labelling of the vertices of the $\cG_N$'s, which is, however, to some extent arbitrary. Changing the way in which the vertices are labelled clearly does not alter the graph $\cG_N$ or the topology of its connections. Nevertheless, it may significantly affect the adjacency matrix $\bA_N$ and consequently the corresponding random-free $N$-step graphon $W_N$. In turn, this may impact on the convergence of the sequence $\{W_N\}_{N\in\N}$ and on the structure of the limit graphon $W$. To remove this source of arbitrariness, it is useful to introduce another metric, derived from that induced by the cut norm, known as the \textit{cut distance} $\delta_\Box$, see~\cite{lovasz2012BOOK} and also~\cite[Theorem~6.9(vi)]{janson2013NYJM}:
$$ \delta_\Box(W_N,W):=\inf_{\tilde{\sigma}}\norm{W-W_N^{\tilde{\sigma}}}{\Box}, $$
where $W_N^{\tilde{\sigma}}(x,x_\ast):=W_N(\tilde{\sigma}(x),\tilde{\sigma}(x_\ast))$ and $\tilde{\sigma}:[0,\,1]\to [0,\,1]$ is a (Lebesgue) measure-preserving bijection built from a permutation $\sigma:\{1,\,\dots,\,N\}\to\{1,\,\dots,\,N\}$ in such a way that $\tilde{\sigma}(I_i)=I_{\sigma(i)}$, the $I_i$'s being the intervals defined in~\eqref{eq:Ii}. In practice, one quotients the $W_N$'s by permutations of the mesh $\{I_i\}_{i=1}^{N}$, thereby identifying as equivalent all random-free $N$-step graphons that differ only by a rearrangement of their pixels. Then, one takes the infimum of their cut-norm distances from $W$, which indeed defines a norm on the quotient space.

Notice that
\begin{align*}
    W_N^{\tilde{\sigma}}(x,x_\ast) &= \sum_{i=1}^{N}\sum_{j=1}^{N}a^N_{ij}\chi_{I_i}(\tilde{\sigma}(x))\chi_{I_j}(\tilde{\sigma}(x_\ast)) \\
    &= \sum_{i=1}^{N}\sum_{j=1}^{N}a^N_{ij}\chi_{\tilde{\sigma}^{-1}(I_i)}(x)\chi_{\tilde{\sigma}^{-1}(I_j)}(x_\ast)
\end{align*}
and that $\tilde{\sigma}^{-1}(I_i)$ is, by construction of $\tilde{\sigma}$, one of the intervals~\eqref{eq:Ii} for every $i=1,\,\dots,\,N$. Thus, also $W_N^{\tilde{\sigma}}$ is a random-free $N$-step graphon on the same pixelation as that of $W_N$. This means that~\eqref{eq:L1<=cut} applies also to the $W_N^{\tilde{\sigma}}$'s, yielding
\begin{equation}
    \inf_{\tilde{\sigma}}\norm{W-W_N^{\tilde{\sigma}}}{L^1([0,\,1]^2)}\leq\sqrt{2N}\delta_\Box(W,W_N).
    \label{eq:inf_L1<=delta_cut}
\end{equation}

\subsection{Examples of graphons}
\label{sect:ex_graphons}
\begin{figure}[!t]
\centering
\includegraphics[width=.8\linewidth]{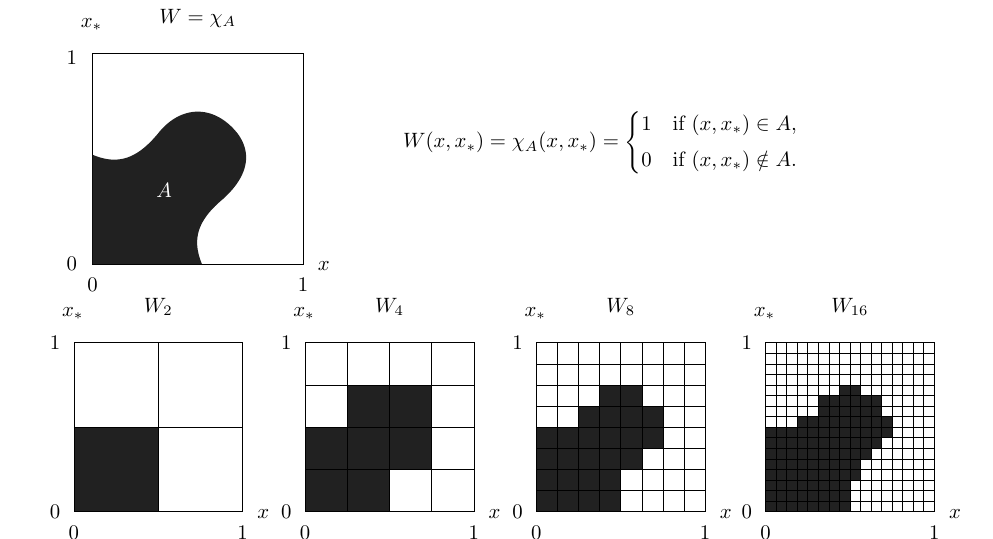}
\caption{Approximation of the graphon $W(x,x_\ast)=\chi_A(x,x_\ast)$. Black pixels correspond to $1$, white pixels correspond to $0$}
\label{fig:det_graphon}
\end{figure}

We now present examples of graphons $W:[0,\,1]^2\to [0,\,1]$ that arise as limits, in the cut distance $\delta_\Box$, of random-free $N$-step graphons $W_N:[0,\,1]^2\to\{0,\,1\}$, and discuss their potential modelling interpretations.

To begin, we consider a random-free graphon $W:[0,\,1]^2\to\{0,\,1\}$, that is, a graphon taking only the values $0$ and $1$, obtained as the limit of random-free $N$-steps graphons. Let $A\subseteq [0,\,1]^2$ be a measurable subset of the unit square with smooth boundary $\partial A$. For $(x,\,x_\ast)\in I_i\times I_j$, define
$$ W_N(x,x_\ast)=
    \begin{cases}
        1 & \text{if } (I_i\times I_j)\cap A\neq\emptyset \\
        0 & \text{otherwise},
    \end{cases}
    \qquad
    i,\,j=1,\,\dots,\,N, $$
so that $W_N$ takes the value $1$ on those cells $I_i\times I_j$ that intersect $A$, and $0$ on all remaining cells. We claim that $W_N$ converges to the random-free graphon $W(x,x_\ast)=\chi_A(x,x_\ast)$.

Indeed, $W_N$ and $W$ coincide on all cells that are either entirely contained in $A$ or entirely contained in $[0,\,1]^2\setminus A$: by construction, both take the value $1$ on the former and the value $0$ on the latter. The two graphons differ only on cells intersecting $\partial A$. The $L^1$-distance between them can therefore be estimated by the area of a tubular neighbourhood of $\partial A$, whose length is that of $\partial A$ and whose width is bounded from above by the length of the diagonal of a cell, namely $\frac{\sqrt{2}}{N}$. Consequently,
$$ \delta_\Box(W_N,W)\leq\norm{W-W_N}{\Box}\leq\norm{W-W_N}{L^1([0,\,1]^2)}\leq \frac{C}{N} $$
for some constant $C>0$, whence $\delta_\Box(W_N,W)\to 0$ as $N\to\infty$. Figure~\ref{fig:det_graphon} illustrates the convergence.

From a modelling perspective, the graphon $W(x,x_\ast)=\chi_A(x,x_\ast)$ represents a fully connected and isolated community $A$ within a larger multi-agent system. Indeed, every member of $A$ is connected to every other member, while having no connections to agents outside $A$. Depending on the structure of $A$, the latent variable $x$ may admit different interpretations. For instance, it may encode \textit{homophily} if $A$ consists predominantly of pairs $(x,\,x_\ast)$ with similar values of $x$ and $x_\ast$, as in the case $A=\{(x,\,x_\ast)\in [0,\,1]^2\,:\,\abs{x_\ast-x}\leq r\}$ for some threshold $r>0$. Alternatively, it may encode \textit{hierarchy} if $A$ consists mainly of pairs $(x,\,x_\ast)$ whose values differ substantially, as in the case $A=\{(x,\,x_\ast)\in [0,\,1]^2\,:\,\abs{x_\ast-x}\geq r\}$ for some threshold $r>0$. For $W$ to be symmetric, the set $A$ must itself be symmetric with respect to the diagonal $x_\ast=x$ of the unit square.

As a second example, let $W:[0,\,1]^2\to [0,\,1]$ be a Lipschitz continuous graphon and consider a \textit{weighted} graph $\cG_N$, whose adjacency matrix $\bA_N$ has entries
$$ a^N_{ij}=W\!\left(\frac{i}{N},\frac{j}{N}\right). $$
This provides a natural extension of the classical $0-1$ adjacency matrix of Section~\ref{sect:kin_eq.graph} and does not affect the subsequent theory. The quantity $a^N_{ij}\in [0,\,1]$ is typically interpreted as the probability that an edge exists between the vertices $i,\,j\in\cV_N$. We claim that, as $N\to\infty$, the $N$-step graphons\footnote{Note that these $W_N$ are generally \textit{not} random-free.} $W_N$ associated with this adjacency matrix as in~\eqref{eq:WN} converge to $W$ in the cut distance.

Indeed, we have:
\begin{align*}
    \norm{W-W_N}{L^1([0,\,1]^2)} &= \sum_{i=1}^{N}\sum_{j=1}^{N}\int_{I_j}\int_{I_i}\abs{W(x,x_\ast)-W_N(x,x_\ast)}\,dx\,dx_\ast \\
    &= \sum_{i=1}^{N}\sum_{j=1}^{N}\int_{I_j}\int_{I_i}\abs*{W(x,x_\ast)-a^N_{ij}}\,dx\,dx_\ast \\
    &= \sum_{i=1}^{N}\sum_{j=1}^{N}\int_{I_j}\int_{I_i}\abs*{W(x,x_\ast)-W\!\left(\frac{i}{N},\frac{j}{N}\right)}\,dx\,dx_\ast \\
    &\leq L_W\sum_{i=1}^{N}\sum_{j=1}^{N}\int_{I_j}\int_{I_i}\left(\abs*{x-\frac{i}{N}}+\abs*{x_\ast-\frac{j}{N}}\right)\,dx\,dx_\ast,
\end{align*}
where $L_W>0$ is the Lipschitz constant of $W$. Since, for $x\in I_i$ and $x_\ast\in I_j$, we have $\abs{x-\frac{i}{N}}\leq\frac{1}{N}$ and $\abs{x_\ast-\frac{j}{N}}\leq\frac{1}{N}$, and since the pixelation of the unit square consists of $N^2$ pixels, each of area $\frac{1}{N^2}$, it follows that
$$ \norm{W-W_N}{L^1([0,\,1]^2)}\leq L_W\cdot\frac{2}{N}\cdot\frac{1}{N^2}\cdot N^2=\frac{2L_W}{N}. $$
We therefore obtain the desired convergence in the $L^1$-norm and, consequently, in the cut distance as well.

\begin{figure}[!t]
\centering
\includegraphics[width=.8\linewidth]{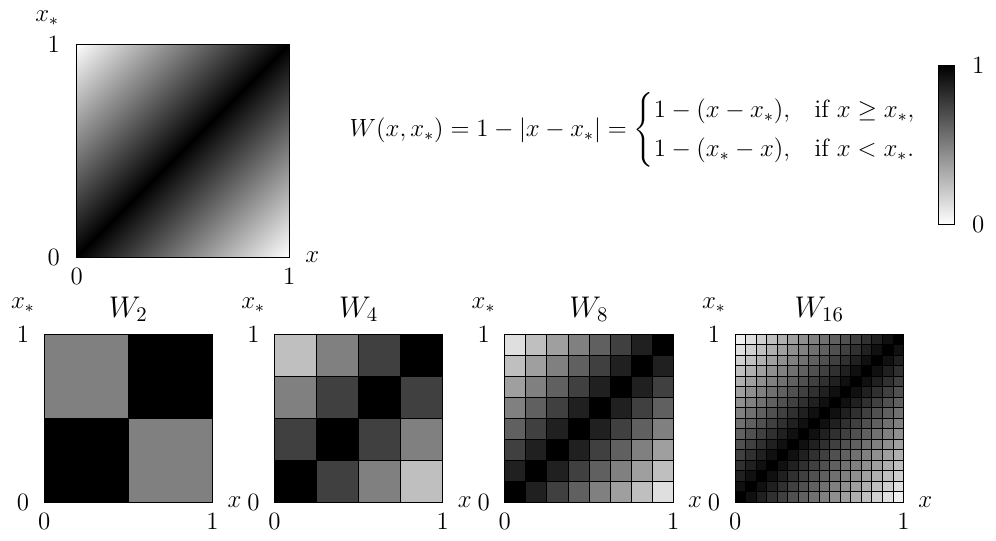}
\caption{Approximation of the graphon $W(x,x_\ast)=1-\abs{x-x_\ast}$}
\label{fig:stoch_graphon}
\end{figure}

Figure~\ref{fig:stoch_graphon} illustrates this convergence for the graphon $W(x,x_\ast)=1-\abs{x-x_\ast}$, which models a homophilic community in which agents with similar values of the latent variable are more likely to be connected. Such a graphon arises, for instance, in scientific collaboration networks, where $x$ represents an agent's expertise along a given disciplinary spectrum; in simplified geographical models, where $x$ denotes an agent's position on a line; and in preference-based models, where $x$ may represent an agent's opinion or a consumer's personal taste.

Convergence in the cut distance to a graphon $W:[0,\,1]^2\to [0,\,1]$ which takes intermediate values between $0$ and $1$ can also be achieved, in probability, starting from a \textit{random} graph, i.e. one whose adjacency matrix's entries are Bernoulli random variables:
$$ a^N_{ij}\sim\operatorname{Bernoulli}\!\left(W\!\left(\frac{i}{N},\frac{j}{N}\right)\right). $$
This means that an edge between vertices $i,\,j$ exists with probability $W(\frac{i}{N},\frac{j}{N})$. In this case, the $W_N$'s in~\eqref{eq:WN} are random variables and their realisations are random-free $N$-step graphons. If $W$ is smooth enough, for instance Lipschitz continuous, then it can be shown that $\delta_\Box(W_N,W)\to 0$ almost surely as $N\to\infty$, cf. e.g.,~\cite{lovasz2012BOOK}.

\subsection{Kinetic equations on finite graphs}
In the classical kinetic theory for multi-agent systems~\cite{loy2026RMUP,pareschi2013BOOK}, homogeneous Boltzmann-type equations modelling pairwise interactions among agents structured by a \textit{microscopic trait} $v\in\R$ are written as
\begin{equation}
    \frac{\partial f}{\partial t}(v,t)=\int_\R\ave*{\frac{1}{J}f(\pr{v},t)f(\pr{v}_\ast,t)-f(v,t)f(v_\ast,t)}\,dv_\ast,
    \label{eq:Boltz-type.strong}
\end{equation}
where $f:\R\times [0,\,+\infty)\to\R_+$ is the time-varying \textit{distribution function} of the trait and $\pr{v},\,\pr{v}_\ast\in\R$ are the \textit{pre-interaction} traits, which generate the \textit{post-interaction} traits $v,\,v_\ast\in\R$ in accordance with a prescribed transformation $(\pr{v},\,\pr{v}_\ast)\mapsto (v,\,v_\ast)$, called the \textit{interaction rule}, with Jacobian factor $J>0$. Since the interaction rule may depend on random parameters, in~\eqref{eq:Boltz-type.strong} one also takes the expectation $\ave{\cdot}$ with respect to the laws of such parameters.

More often, homogeneous Boltzmann-type equations are written in weak form by introducing an \textit{observable} $\varphi:\R\to\R$, i.e. any function of the trait $v$, which plays the role of a test function, and integrating~\eqref{eq:Boltz-type.strong} over $v$ after multiplying both sides by $\varphi$:
\begin{equation}
    \frac{d}{dt}\int_\R\varphi(v)f(v,t)\,dv=\int_\R\int_\R\ave{\varphi(v')-\varphi(v)}f(v,t)f(v_\ast,t)\,dv\,dv_\ast,
    \label{eq:Boltz-type.weak}
\end{equation}
where now $v'\in\R$ denotes the post-interaction trait obtained from the pre-interaction traits $v,\,v_\ast\in\R$ via the interaction rule. In particular, we consider the latter in the \textit{symmetric} form
\begin{equation}
    v'=\Psi(v,v_\ast,\omega), \qquad v_\ast'=\Psi(v_\ast,v,\omega),
    \label{eq:int_rules}
\end{equation}
meaning that either rule is obtained from the other by swapping $v$ and $v_\ast$. Here, $\Psi:\R\times\R\times\Omega\to\R$ is a function of the pre-interaction traits and of an auxiliary variable $\omega$ ranging in a sample space $\Omega$, which is introduced to account for the possible presence of random parameters in the interaction rule as mentioned above. Accordingly,
$$ \ave{\cdot}:=\int_\Omega(\cdot)\,dP(\omega), $$
where $P$ is a probability measure defined on the measurable space consisting of $\Omega$ together with some suitable $\sigma$-algebra of its subsets.

A tacit assumption in this setting is that interactions are ``all-to-all'', i.e. that any pair of agents sampled from the available pool of individuals may possibly interact. In~\cite{nurisso2024EJAM}, instead, agents have been regarded as the vertices of a finite graph $\cG_N=(\cV_N,\,\cE_N)$, whose edges connect those among them that may actually interact. The homogeneous Boltzmann-type equation generalising~\eqref{eq:Boltz-type.strong} within this new perspective is
\begin{equation}
    \frac{\partial f_i}{\partial t}(v,t)=\frac{1}{N}\sum_{j=1}^{N}a^N_{ij}\int_\R\ave*{\frac{1}{J}f_i(\pr{v},t)f_j(\pr{v}_\ast,t)
        -f_i(v,t)f_j(v_\ast,t)}\,dv_\ast,
    \label{eq:Boltz-type.strong.graph}
\end{equation}
for $i=1,\,\dots,\,N$, which in weak form reads
\begin{equation}
    \frac{d}{dt}\int_\R\varphi(v)f_i(v,t)\,dv=\frac{1}{N}\sum_{j=1}^{N}a^N_{ij}\int_\R\int_\R\ave{\varphi(v')-\varphi(v)}f_i(v,t)f_j(v_\ast,t)\,dv\,dv_\ast.
    \label{eq:Boltz-type.weak.graph}
\end{equation}
Here, $f_i:\R\times [0,\,+\infty)\to\R_+$ is the distribution function of the trait of the $i$-th agent and the $a^N_{ij}$'s are the entries of the adjacency matrix $\bA_N$ of $\cG_N$. In particular, each $f_i$ is taken as a probability density function (or, more generally, a probability measure), therefore it satisfies the normalisation condition
\begin{equation}
    \int_\R f_i(v,t)\,dv=1, \qquad \forall\,t\geq 0, \quad \forall\,i=1,\,\dots,\,N.
    \label{eq:fi.int}
\end{equation}

The system of equations~\eqref{eq:Boltz-type.strong.graph}, or equivalently~\eqref{eq:Boltz-type.weak.graph}, has been formally deduced in~\cite{nurisso2024EJAM} from an underlying stochastic particle system whose interaction rates are proportional to the entries of $\bA_N$. For completeness, we point out that in~\cite{nurisso2024EJAM} the more general case of a directed graph with non-symmetric adjacency matrix, together with non-symmetric interaction rules, is considered. Systems~\eqref{eq:Boltz-type.strong.graph},~\eqref{eq:Boltz-type.weak.graph} are what the equations deduced there reduce to under the assumptions of undirected graph and symmetric interaction rules.

We now rewrite~\eqref{eq:Boltz-type.weak.graph} using the random-free $N$-step graphon $W_N$. For this, we introduce
$$ f_N(x,v,t):=\sum_{i=1}^{N}f_i(v,t)\chi_{I_i}(x), $$
which defines a piecewise constant extension in $x$ of the $f_i$'s over the partition of $[0,\,1]$ given by the subintervals $I_i$ in~\eqref{eq:Ii}. We observe that such an $f_N$ is a probability distribution function with respect to the variable $v\in\R$, for each fixed $x\in [0,\,1]$. Indeed:
$$ \int_\R f_N(x,v,t)\,dv=\sum_{i=1}^{N}\left(\int_\R f_i(v,t)\,dv\right)\chi_{I_i}(x)=\chi_{[0,\,1]}(x)\equiv 1,
    \qquad \forall\,x\in [0,\,1],\,t\geq 0. $$

Next, we observe that for every $\phi:[0,\,1]\to\R$ it results
$$ \int_0^1\phi(x)f_N(x,v,t)\,dx=\sum_{i=1}^{N}f_i(v,t)\int_{I_i}\phi(x)\,dx $$
and also
$$ \int_0^1\int_0^1\phi(x)W_N(x,x_\ast)f_N(x,v,t)f_N(x_\ast,v_\ast,t)\,dx\,dx_\ast
    =\frac{1}{N}\sum_{i=1}^{N}\sum_{j=1}^{N}a^N_{ij}f_i(v,t)f_j(v_\ast,t)\int_{I_i}\phi(x)\,dx, $$
where we have used the expression~\eqref{eq:WN} of $W_N$. Bearing these two relationships in mind, we multiply both sides of~\eqref{eq:Boltz-type.weak.graph} by $\int_{I_i}\phi(x)\,dx$ and sum over $i$ to discover:
\begin{multline*}
    \frac{d}{dt}\int_0^1\int_\R\phi(x)\varphi(v)f_N(x,v,t)\,dv\,dx \\
    =\int_0^1\int_0^1\int_\R\int_\R\phi(x)W_N(x,x_\ast)\ave{\varphi(v')-\varphi(v)}f_N(x,v,t)f_N(x_\ast,v_\ast,t)\,dv\,dv_\ast\,dx\,dx_\ast,
\end{multline*}
which, owing to the arbitrariness of $\phi$, can be rewritten as
\begin{multline}
    \frac{\partial}{\partial t}\int_\R\varphi(v)f_N(x,v,t)\,dv \\
    =\int_0^1\int_\R\int_\R W_N(x,x_\ast)\ave{\varphi(v')-\varphi(v)}f_N(x,v,t)f_N(x_\ast,v_\ast,t)\,dv\,dv_\ast\,dx_\ast.
    \label{eq:Boltz-type.WN}
\end{multline}
We regard this latter equation as the continuous-in-$x$ analogue of~\eqref{eq:Boltz-type.weak.graph}.

\subsection{The dense graph limit of Boltzmann-type equations}
In this section, we perform the so-called \textit{dense graph limit} of~\eqref{eq:Boltz-type.WN}: we consider a sequence $\{\cG_N\}_{N\in\N}$ of growing graphs -- where ``growing'' means $N\to\infty$, being $N$ the cardinality of the set of vertices $\cV_N$ -- and we study the limit behaviour of~\eqref{eq:Boltz-type.WN} under the assumption that the sequence $\{W_N\}_{N\in\N}$ of random-free $N$-step graphons which encode the $\bA_N$'s converges to some graphon $W$ in the sense of the cut distance $\delta_\Box$.

\subsubsection{Linear interaction rules}
\label{sect:lin_int}
As a prototypical case, which is both amenable to detailed analytical investigation and relevant for applications, we consider \textit{linear} interaction rules featuring a function $\Psi$ in~\eqref{eq:int_rules} of the form
\begin{equation}
    \Psi(v,v_\ast,\omega)=p(\omega)v+q(\omega)v_\ast,
    \label{eq:lin_int}
\end{equation}
where $p,\,q\geq 0$, with\footnote{This is introduced to guarantee that $p,\,q$ are not simultaneously zero.} $pq>0$, are, in general, random interaction coefficients with prescribed probability laws. Rules of this type include, for instance, most one-dimensional Boltzmann-type models for Maxwell molecules available in the literature, among them the celebrated Kac model~\cite{kac1959BOOK}. More recently, Boltzmann-type sociophysical models have also been formulated by means of linear symmetric interaction rules, with agent's traits defined on the whole real line; see e.g.,~\cite{loy2020CMS}.

In the following, for brevity, we will omit the explicit dependence of $p$ and $q$ on $\omega$.

Homogeneous Boltzmann-type equations with linear interactions may be conveniently approached by means of a particular class of metrics for probability measures, known as \textit{Fourier metrics}. They are based on the Fourier transform of probability measures, which, as first noticed by Bobylev~\cite{bobylev1975DANSSSR}, confers a particularly tractable structure on the homogeneous Boltzmann equation for Maxwellian molecules.

Let $\cP(\R)$ be the space of probability measures defined on the Borel $\sigma$-algebra of $\R$. Moreover, for $s\in\N$, $s>0$, let
$$ \cP_s(\R):=\left\{\mu\in\cP(\R)\,:\,\int_\R\abs{v}^s\,d\mu(v)<+\infty\right\}. $$
Given $\mu\in\cP(\R)$, the Fourier transform of $\mu$ is the bounded and continuous function $\hat{\mu}:\R\to\mathbb{C}$,
$$ \hat{\mu}(\xi):=\int_\R e^{-i\xi v}\,d\mu(v). $$
If $\nu\in\cP(\R)$ is another probability measure, the \textit{$s$-Fourier distance} between $\mu$ and $\nu$ is
$$ d_s(\mu,\nu):=\sup_{\xi\in\R\setminus\{0\}}\frac{\abs{\hat{\nu}(\xi)-\hat{\mu}(\xi)}}{\abs{\xi}^s}. $$
Notice that $\hat{\mu}(0)=\hat{\nu}(0)=1$, therefore the quantity $\frac{\abs{\hat{\nu}(\xi)-\hat{\mu}(\xi)}}{\abs{\xi}^s}$ need not be unbounded for $\xi\to 0$. In more detail, as stated in~\cite{carrillo2007RMUP,loy2026RMUP},
\begin{lemma} \label{lemma:ds.finite}
Let $\mu,\,\nu\in\cP_s(\R)$, $s\in\N$. If $\mu,\,\nu$ have the same moments at least up to the order $s-1$, i.e. if
$$ \int_\R v^n\,d\mu(v)=\int_\R v^n\,d\nu(v), \qquad \forall\,n\in\N,\ n\leq s-1, $$
then $d_s(\mu,\nu)<+\infty$.
\end{lemma}

In particular, from Lemma~\ref{lemma:ds.finite} with $s=1$ we deduce that the $1$-Fourier distance $d_1$ between any two probability measures $\mu,\,\nu\in\cP_1(\R)$ is finite because $\mu,\,\nu$ have clearly the same zeroth-order moment.

In the following, we study the dense graph limit of~\eqref{eq:Boltz-type.WN} with linear interaction rules~\eqref{eq:lin_int}, taking advantage of the $1$-Fourier distance. To this end, we first prove that
$$ f_N(x,\cdot,t)\in\cP_1(\R), \qquad \forall\,x\in [0,\,1],\,t>0, $$
i.e. that the solutions to~\eqref{eq:Boltz-type.WN} belong to $\cP_1(\R)$ as probability measures in $v\in\R$, keeping $x\in [0,\,1]$ and $t>0$ fixed.

\begin{lemma} \label{lemma:m1N}
Let
$$ \cm_{1,N}(x,t):=\int_\R\abs{v}f_N(x,v,t)\,dv, $$
$f_N$ being a solution to~\eqref{eq:Boltz-type.WN}-\eqref{eq:lin_int}. Then
$$ \norm{\cm_{1,N}(t)}{\infty}\leq\norm{\cm_{1,N}^0}{\infty}e^{\ave{p+q}t}, \qquad \forall\,t>0, $$
where $\cm_{1,N}^0(x)=\cm_{1,N}(x,0)$.
\end{lemma}
\begin{proof}
We take $\varphi(v)=\abs{v}$ in~\eqref{eq:Boltz-type.WN} to deduce that $\cm_{1,N}$ satisfies
\begin{align*}
    \frac{\partial\cm_{1,N}}{\partial t}(x,t) &=
        \int_0^1\int_\R\int_\R W_N(x,x_\ast)\ave{\abs{v'}-\abs{v}}f_N(x,v,t)f_N(x_\ast,v_\ast,t)\,dv\,dv_\ast\,dx_\ast \\
    &\leq \int_0^1W_N(x,x_\ast)\bigl(\ave{p-1}\cm_{1,N}(x,t)+\ave{q}\cm_{1,N}(x_\ast,t)\bigr)\,dx_\ast \\
    &\leq \ave{p+q}\cd_{W_N}(x)\norm{\cm_{1,N}(t)}{\infty}-\cd_{W_N}(x)\cm_{1,N}(x,t).
\intertext{Moving the term $-\cd_{W_N}(x)\cm_{1,N}(x,t)$ to the left-hand side, multiplying both sides by $e^{\cd_{W_N}(x)t}$, and integrating then in time over $[0,\,t]$, $t>0$, yields}
    e^{\cd_{W_N}(x)t}\cm_{1,N}(x,t) &\leq \norm{\cm_{1,N}^0}{\infty}
        +\ave{p+q}\norm{\cd_{W_N}}{\infty}\int_0^te^{\cd_{W_N}(x)\tau}\norm{\cm_{1,N}(\tau)}{\infty}\,d\tau.
\end{align*}
Now, multiplying both sides by $e^{-\cd_{W_N}(x)t}$ we observe that, on the right-hand side, we can take advantage of the bounds $e^{-\cd_{W_N}(x)t}\leq e^{-\dWNinf t}$ for every $t\geq 0$ and $e^{\cd_{W_N}(x)(\tau-t)}\leq e^{\dWNinf(\tau-t)}$ for every $\tau\in [0,\,t]$. Thus,
$$ \cm_{1,N}(x,t)\leq\norm{\cm_{1,N}^0}{\infty}e^{-\dWNinf t}
    +\ave{p+q}\norm{\cd_{W_N}}{\infty}\int_0^te^{\dWNinf(\tau-t)}\norm{\cm_{1,N}(\tau)}{\infty}\,d\tau, $$
and consequently
$$ e^{\dWNinf t}\norm{\cm_{1,N}(t)}{\infty}\leq\norm{\cm_{1,N}^0}{\infty}
    +\ave{p+q}\norm{\cd_{W_N}}{\infty}\int_0^te^{\dWNinf\tau}\norm{\cm_{1,N}(\tau)}{\infty}\,d\tau. $$
At this point, applying Gr\"{o}nwall's inequality to the function $e^{\dWNinf t}\norm{\cm_{1,N}(t)}{\infty}$ yields
$$ \norm{\cm_{1,N}(t)}{\infty}\leq\norm{\cm_{1,N}^0}{\infty}e^{\left(\ave{p+q}\norm{\cd_{W_N}}{\infty}-\dWNinf\right)t}, $$
whence the thesis follows using $\dWNinf\geq 0$ and $\norm{\cd_{W_N}}{\infty}\leq 1$.
\end{proof}

Lemma~\ref{lemma:m1N} gives the finiteness of $\cm_{1,N}(x,t)$, thereby implying that $f_N(x,\cdot,t)\in\cP_1(\R)$, for every $x\in [0,\,1]$ and every $t>0$ provided the same holds at $t=0$.

We are now in a position to prove the dense graph limit of~\eqref{eq:Boltz-type.WN},~\eqref{eq:lin_int}.
\begin{theorem}[Dense graph limit with linear interactions] \label{theo:dense_graph.linear}
Let $f_N^0(x,v):=f_N(x,v,0)$ and assume that the sequence $\{f_N^0(x,\cdot)\}_{N\in\N}\subset\cP_1(\R)$ satisfies the following properties:
\begin{enumerate}[label=(\roman*)]
\item \label{ass:m10N.unif_bound} the $\cm_{1,N}^0$'s are uniformly bounded with respect to $N$, i.e. there exists a constant $c>0$ such that $\norm{\cm_{1,N}^0}{\infty}\leq c$ for every $N\in\N$;
\item \label{ass:conv.f_N^0->f^0.d1} $\{f_N^0(x,\cdot)\}_{N\in\N}$ converges, for $N\to\infty$, to some $f^0=f^0(x,\cdot)\in\cP_1(\R)$ in the $1$-Fourier distance, i.e.
$$ \lim_{N\to\infty}d_1(f_N^0(x),f^0(x))=0, \qquad \text{for a.e. } x\in [0,\,1]. $$
\end{enumerate}
Assume moreover that the sequence $\{W_N\}_{N\in\N}$ converges, for $N\to\infty$, to some $W\in L^\infty([0,\,1]^2)$ in the cut distance with a rate of convergence higher than $N^{-1/2}$:
$$ \delta_\Box(W_N,W)=o\!\left(\frac{1}{\sqrt{N}}\right) \quad \text{as } N\to\infty. $$
Then the solutions of~\eqref{eq:Boltz-type.WN},~\eqref{eq:lin_int} converge, in the $1$-Fourier distance, to the solutions of the following limit equation:
\begin{equation}
    \frac{\partial}{\partial t}\int_\R\varphi(v)f(x,v,t)\,dv
        =\int_0^1\int_\R\int_\R W(x,x_\ast)\ave{\varphi(v')-\varphi(v)}f(x,v,t)f(x_\ast,v_\ast,t)\,dv\,dv_\ast\,dx_\ast
    \label{eq:Boltz-type.W}
\end{equation}
emanating from $f^0$, where $\varphi$ is an arbitrary observable and $v'=pv+qv_\ast$ as in~\eqref{eq:int_rules},~\eqref{eq:lin_int}.
\end{theorem}
\begin{proof}
To begin with, we observe that, letting $\varphi(v)=1$ in~\eqref{eq:Boltz-type.W}, we have
$$ \frac{\partial}{\partial t}\int_\R f(x,v,t)\,dv=0, $$
whence we deduce that also $f$ is a probability distribution function with respect to the variable $v$, for each fixed $x\in [0,\,1]$ and $t>0$, as it is so at the initial time $t=0$ by assumption. Moreover,
since the limit equation~\eqref{eq:Boltz-type.W} has the same structure as that of~\eqref{eq:Boltz-type.WN}, owing to Lemma~\ref{lemma:m1N} we deduce that $f(x,\cdot,t)\in\cP_1(\R)$ for every $x\in [0,\,1]$ and every $t>0$ if $\norm{\cm_1^0}{\infty}<+\infty$.

We will show that, in the stated assumptions, $f_N$ converges to $f$ by proving, in particular, that $d_1(f_N(x,t),f(x,t))\to 0$ when $N\to\infty$ for a.e. $x\in [0,\,1]$ and every $t>0$.

For this, it is convenient to Fourier-transform both~\eqref{eq:Boltz-type.WN} and~\eqref{eq:Boltz-type.W} by choosing $\varphi(v)=e^{-i\xi v}$. If
$$ \hat{f}_N(x,\xi,t):=\int_\R f_N(x,v,t)e^{-i\xi v}\,dv, \qquad \hat{f}(x,\xi,t):=\int_\R f(x,v,t)e^{-i\xi v}\,dv $$
denote the $v$-Fourier transforms of $f_N,\,f$, respectively, we get:
\begin{equation}
    \frac{\partial\hat{f}_N}{\partial t}(x,\xi,t)
        =\int_0^1W_N(x,x_\ast)\ave{\hat{f}_N(x,p\xi,t)\hat{f}_N(x_\ast,q\xi,t)}\,dx_\ast-\cd_{W_N}(x)\hat{f}_N(x,\xi,t),
    \label{eq:Boltz-type.WN.Fourier}
\end{equation}
and likewise
\begin{equation}
    \frac{\partial\hat{f}}{\partial t}(x,\xi,t)
        =\int_0^1W(x,x_\ast)\ave{\hat{f}(x,p\xi,t)\hat{f}(x_\ast,q\xi,t)}\,dx_\ast-\cd_W(x)\hat{f}(x,\xi,t).
    \label{eq:Boltz-type.W.Fourier}
\end{equation}

Subtracting~\eqref{eq:Boltz-type.WN.Fourier} from~\eqref{eq:Boltz-type.W.Fourier}, dividing by $\abs{\xi}$, and rearranging the terms we obtain:
\begin{align*}
    \frac{\partial}{\partial t}&\frac{\hat{f}(x,\xi,t)-\hat{f}_N(x,\xi,t)}{\abs{\xi}} \\
    &= \int_0^1W(x,x_\ast)\ave*{\frac{\hat{f}(x,p\xi,t)-\hat{f}_N(x,p\xi,t)}{\abs{\xi}}\hat{f}(x_\ast,q\xi,t)}\,dx_\ast \\
    &\phantom{=} +\int_0^1W_N(x,x_\ast)\ave*{\hat{f}_N(x,p\xi,t)\frac{\hat{f}(x_\ast,q\xi,t)-\hat{f}_N(x_\ast,q\xi,t)}{\abs{\xi}}}\,dx_\ast \\
    &\phantom{=} +\int_0^1\left(W(x,x_\ast)-W_N(x,x_\ast)\right)
        \ave*{\frac{\hat{f}_N(x,p\xi,t)\hat{f}(x_\ast,q\xi,t)-\hat{f}(x,\xi,t)}{\abs{\xi}}}\,dx_\ast \\
    &\phantom{=} -\cd_{W_N}(x)\frac{\hat{f}(x,\xi,t)-\hat{f}_N(x,\xi,t)}{\abs{\xi}}.
\end{align*}
Now, moving the last term on the right-hand to the left-hand side and multiplying both sides by $e^{\cd_{W_N}(x)t}$ we deduce:
\begin{align*}
    \frac{\partial}{\partial t} &\left(e^{\cd_{W_N}(x)t}\frac{\hat{f}(x,\xi,t)-\hat{f}_N(x,\xi,t)}{\abs{\xi}}\right) \\
    &= e^{\cd_{W_N}(x)t}\int_0^1W(x,x_\ast)\ave*{p\frac{\hat{f}(x,p\xi,t)-\hat{f}_N(x,p\xi,t)}{\abs{p\xi}}\hat{f}(x_\ast,q\xi,t)}\,dx_\ast \\
    &\phantom{=} +e^{\cd_{W_N}(x)t}\int_0^1W_N(x,x_\ast)
        \ave*{q\hat{f}_N(x,p\xi,t)\frac{\hat{f}(x_\ast,q\xi,t)-\hat{f}_N(x_\ast,q\xi,t)}{\abs{q\xi}}}\,dx_\ast \\
    &\phantom{=} +e^{\cd_{W_N}(x)t}\int_0^1\left(W(x,x_\ast)-W_N(x,x_\ast)\right)
         \ave*{\frac{\hat{f}_N(x,p\xi,t)\hat{f}(x_\ast,q\xi,t)-\hat{f}(x,\xi,t)}{\abs{\xi}}}\,dx_\ast.
\end{align*}
From here, taking the absolute value of both sides and recalling that $\abs{\partial_t(\cdot)}\geq\partial_t\abs{\cdot}$ and also that $\abs{\hat{f}},\,\abs{\hat{f}_N}\leq 1$ (as a general property of the Fourier transform of probability measures), we discover:
\begin{align}
    \begin{aligned}[b]
    \frac{\partial}{\partial t} &\left(e^{\cd_{W_N}(x)t}\frac{\abs{\hat{f}(x,\xi,t)-\hat{f}_N(x,\xi,t)}}{\abs{\xi}}\right) \\
    &\leq e^{\cd_{W_N}(x)t}\Biggl[
        \norm{\cd_W}{\infty}\ave{p}d_1(f_N(x,t),f(x,t)) \\
    &\phantom{\leq} +\ave{q}\int_0^1W_N(x,x_\ast)d_1(f_N(x_\ast,t),f(x_\ast,t))\,dx_\ast \\
    &\phantom{\leq} +\int_0^1\abs{W(x,x_\ast)-W_N(x,x_\ast)}
        \ave*{\frac{\abs{\hat{f}_N(x,p\xi,t)\hat{f}(x_\ast,q\xi,t)-\hat{f}(x,\xi,t)}}{\abs{\xi}}}\,dx_\ast\Biggr].
    \end{aligned}
    \label{eq:dense_graph_limit.proof}
\end{align}

We now turn to the last term on the right-hand side of~\eqref{eq:dense_graph_limit.proof}. Write
$$ \hat{f}_N(x,p\xi,t)\hat{f}(x_\ast,q\xi,t)-\hat{f}(x,\xi,t)
    =\hat{f}_N(x,p\xi,t)\left(\hat{f}(x_\ast,q\xi,t)-1\right)+\hat{f}_N(x,p\xi,t)-\hat{f}(x,\xi,t) $$
to see that
$$ \frac{\abs{\hat{f}_N(x,p\xi,t)\hat{f}(x_\ast,q\xi,t)-\hat{f}(x,\xi,t)}}{\abs{\xi}}
    \leq\frac{\abs{\hat{f}(x_\ast,q\xi,t)-1}}{\abs{\xi}}+
        \frac{\abs{\hat{f}_N(x,p\xi,t)-\hat{f}(x,\xi,t)}}{\abs{\xi}}. $$
By performing a first-order Taylor expansion of the map $\xi\mapsto e^{-i\xi v}$ about $\xi=0$ with Lagrange remainder, we obtain $e^{-i\xi v}=1-ive^{-i\bar{\xi}v}\xi$ for some $\bar{\xi}=\theta\xi$ with $\theta\in [0,\,1]$. Hence,
\begin{align}
    \begin{aligned}[b]
    \frac{\abs{\hat{f}(x_\ast,q\xi,t)-1}}{\abs{\xi}} &= \frac{1}{\abs{\xi}}\abs*{\int_\R f(x_\ast,v,t)e^{-iq\xi v}\,dv-1} \\
    &= q\abs*{\int_\R vf(x_\ast,v,t)e^{-iq\bar{\xi}v}\,dv}\leq q\cm_1(x_\ast,t)\leq q\norm{\cm_1(t)}{\infty}
    \end{aligned}
    \label{eq:Taylor_exp.proof-1}
\end{align}
and likewise
\begin{align}
    \begin{aligned}[b]
    \frac{\abs{\hat{f}_N(x,p\xi,t)-\hat{f}(x,\xi,t)}}{\abs{\xi}} &=
        \frac{1}{\abs{\xi}}\abs*{\int_\R f_N(x,v,t)e^{-ip\xi v}\,dv-\int_\R f(x,v,t)e^{-i\xi v}\,dv} \\
    &= \abs*{p\int_\R vf_N(x,v,t)e^{-ip\tilde{\xi}v}\,dv-\int_\R vf(x,v,t)e^{-i\bar{\xi}v}\,dv} \\
    &\leq p\cm_{1,N}(x,t)+\cm_1(x,t) \\
    &\leq p\norm{\cm_{1,N}(t)}{\infty}+\norm{\cm_1(t)}{\infty}.
    \end{aligned}
    \label{eq:Taylor_exp.proof-2}
\end{align}
Owing to assumption~\ref{ass:m10N.unif_bound} combined with Lemma~\ref{lemma:m1N} applied to both $f_N$ and $f$, we conclude that there exists $C>0$ such that
$$ \max\{\norm{\cm_{1,N}(t)}{\infty},\,\norm{\cm_1(t)}{\infty}\}\leq Ce^{\ave{p+q}t}, $$
with $C:=\max\{c,\,\norm{\cm_1^0}{\infty}\}$. Therefore,
$$ \ave*{\frac{\abs{\hat{f}_N(x,p\xi,t)\hat{f}(x_\ast,q\xi,t)-\hat{f}(x,\xi,t)}}{\abs{\xi}}}
    \leq C(\ave{p+q}+1)e^{\ave{p+q}t}. $$

Using this, we continue the estimate~\eqref{eq:dense_graph_limit.proof} by integrating both sides over $[0,\,t]$, with $t>0$, and then taking the supremum over $\xi\in\R\setminus\{0\}$ on the left-hand side:
\begin{align*}
    e^{\cd_{W_N}(x)t}d_1(f_N(x,t),f(x,t)) &\leq d_1(f_N^0(x),f^0(x))+\int_0^te^{\cd_{W_N}(x)\tau}\Biggl[
        \norm{\cd_W}{\infty}\ave{p}d_1(f_N(x,\tau),f(x,\tau)) \\
    &\phantom{\leq} +\ave{q}\int_0^1W_N(x,x_\ast)d_1(f_N(x_\ast,\tau),f(x_\ast,\tau))\,dx_\ast \\
    &\phantom{\leq} +C(\ave{p+q}+1)e^{\ave{p+q}\tau}\int_0^1\abs{W(x,x_\ast)-W_N(x,x_\ast)}\,dx_\ast\Biggr]\,d\tau.
\end{align*}
By the same argument as in the proof of Lemma~\ref{lemma:m1N}, we replace $e^{\cd_W(x)t}$ on both sides by $e^{\dWNinf t}$, and then integrate with respect to $x$ over $[0,\,1]$ to deduce:
\begin{align*}
    e^{\dWNinf t}\int_0^1 &d_1(f_N(x,t),f(x,t))\,dx \\ 
    &\leq \int_0^1d_1(f_N^0(x),f^0(x))\,dx \\
    &\phantom{\leq} +C\frac{\ave{p+q}+1}{\ave{p+q}+\dWNinf}\left(e^{\left(\ave{p+q}+\dWNinf\right)t}-1\right)\norm{W-W_N}{L^1([0,\,1]^2)} \\
    &\phantom{\leq} +\ave{p\norm{\cd_W}{\infty}+q\norm{\cd_{W_N}}{\infty}}\int_0^te^{\dWNinf\tau}\int_0^1d_1(f_N(x,\tau),f(x,\tau))\,dx\,d\tau.
\end{align*}
In particular, we observe that, owing to the symmetry of $W_N$, in the term proportional to $\ave{q}$ it results $\int_0^1W_N(x,x_\ast)\,dx=\cd_{W_N}(x_\ast)$.

Applying Gr\"{o}nwall's inequality to the function $t\mapsto e^{\dWNinf t}\int_0^1d_1(f_N(x,t),f(x,t))\,dx$ gives
\begin{multline}
    \int_0^1d_1(f_N(x,t),f(x,t))\,dx \\
    \leq\left(\int_0^1d_1(f_N^0(x),f^0(x))\,dx+\psi(t)\norm{W-W_N}{L^1([0,\,1]^2)}\right)e^{\left(\ave{p\norm{\cd_W}{\infty}+q\norm{\cd_{W_N}}{\infty}}
        -\dWNinf\right)t},
    \label{eq:W-WN.no.sigma}
\end{multline}
where
$$ \psi(t):=C\frac{\ave{p+q}+1}{\ave{p+q}}\left(e^{\left(\ave{p+q}+1\right)t}-1\right)\geq
    C\frac{\ave{p+q}+1}{\ave{p+q}+\dWNinf}\left(e^{\left(\ave{p+q}+\dWNinf\right)t}-1\right). $$
We point out that $\ave{p+q}>0$ since, by assumption, $p,\,q$ are both non-negative and not simultaneously zero.

If $W_N$ is replaced by any of its ``permutations'' $W_N^{\tilde{\sigma}}$, the estimate~\eqref{eq:W-WN.no.sigma} still holds unchanged:
\begin{multline*}
    \int_0^1d_1(f_N(x,t),f(x,t))\,dx \\
    \leq\left(\int_0^1d_1(f_N^0(x),f^0(x))\,dx+\psi(t)\norm{W-W_N^{\tilde{\sigma}}}{L^1([0,\,1]^2)}\right)
        e^{\left(\ave{p\norm{\cd_W}{\infty}+q\norm{\cd_{W_N}}{\infty}}-\dWNinf\right)t},
\end{multline*}
therefore it holds also for the infimum of the right-hand side taken over all $\tilde{\sigma}$. Invoking~\eqref{eq:inf_L1<=delta_cut}, we conclude then
$$ \int_0^1d_1(f_N(x,t),f(x,t))\,dx\leq\left(\int_0^1d_1(f_N^0(x),f^0(x))\,dx+\psi(t)\sqrt{2N}\delta_\Box(W_N,W)\right)e^{\ave{p+q}t}, $$
which, in the stated assumptions, implies
$$ \lim_{N\to\infty}\int_0^1d_1(f_N(x,t),f(x,t))\,dx=0, \qquad \forall\,t>0. $$
In particular, we note that from assumptions~\ref{ass:m10N.unif_bound},~\ref{ass:conv.f_N^0->f^0.d1} and Lemma~\ref{lemma:m1N} it follows, by dominated convergence,
$$ \lim_{N\to\infty}\int_0^1d_1(f_N^0(x),f^0(x))\,dx=\int_0^1\lim_{N\to\infty}d_1(f_N^0(x),f^0(x))\,dx=0, $$
since $d_1(f_N^0(x),f^0(x))\leq\norm{\cm_{1,N}^0}{\infty}+\norm{\cm_1^0}{\infty}\leq c+\norm{\cm_1^0}{\infty}$ by an argument analogous to that underlying~\eqref{eq:Taylor_exp.proof-1},~\eqref{eq:Taylor_exp.proof-2}.

From this, we conclude that, up to subsequences, $\lim_{N\to\infty}d_1(f_N(x,t),f(x,t))=0$ for a.e. $x\in [0,\,1]$ and every $t>0$, which completes the proof.
\end{proof}

\subsubsection{More general interaction rules}
In this section, we extend the dense graph limit result of Theorem~\ref{theo:dense_graph.linear} to interaction rules~\eqref{eq:int_rules} that are not necessarily linear, unlike those in~\eqref{eq:lin_int}. To this end, we first recall another notion of distance between probability measures, which we use to assess the convergence of $f_N$ to $f$, since Fourier methods are not well suited to the analysis of Boltzmann-type kinetic equations with non-linear interactions.

The \textit{$1$-Wasserstein distance} $\cW_1$ between $f_N(x,\cdot,t),\,f(x,\cdot,t)\in\cP_1(\R)$ is
$$ \cW_1(f_N(x,t),f(x,t)):=\sup_{\varphi\in\Lip_1(\R)}\abs*{\int_\R\varphi(v)(f(x,v,t)-f_N(x,v,t))\,dv}, $$
where $\Lip_1(\R)$ is the space of Lipschitz continuous functions on $\R$ with at most unitary Lipschitz constant. In particular, we claim that we can assume $\varphi(0)=0$. Indeed, if this is not the case then, defining $\tilde{\varphi}(v):=\varphi(v)-\varphi(0)$, which is clearly such that $\tilde{\varphi}\in\Lip_1(\R)$ with $\tilde{\varphi}(0)=0$, we obtain
\begin{align*}
    \int_\R\varphi(v)(f(x,v,t)-f_N(x,v,t))\,dv &= \int_\R\left(\tilde{\varphi}(v)+\varphi(0)\right)(f(x,v,t)-f_N(x,v,t))\,dv \\
    &= \int_\R\tilde{\varphi}(v)(f(x,v,t)-f_N(x,v,t))\,dv.
\end{align*}
Moreover, in the definition of $\cW_1$ we can omit the absolute value, for if $\varphi\in\Lip_1(\R)$ then also $-\varphi\in\Lip_1(\R)$ and either $\int_\R\varphi(v)(f(x,v,t)-f_N(x,v,t))\,dv$ or $\int_\R(-\varphi(v))(f(x,v,t)-f_N(x,v,t))\,dv$ is non-negative.

The formulation of the distance $\cW_1$ given above is not the original one, but rather a reformulation based on the so-called Kantorovich-Rubinstein duality~\cite{ambrosio2008BOOK,villani2009BOOK}. For our purposes, however, this representation of $\cW_1$ is the most convenient. The first use of Wasserstein distances for the analysis of kinetic equations goes back to the seminal papers by Dobrushin~\cite{dobrushin1979FAA} and Tanaka~\cite{tanaka1978ZWVG}. Connections between Fourier and Wasserstein distances are studied in detail in~\cite{carrillo2007RMUP}.

To effectively handle interaction rules more general than the linear ones, we need to restrict the class of functions $\Psi$ used in~\eqref{eq:int_rules}. Specifically, we assume:
\begin{enumerate}[label=(H\arabic*)]
\item \label{ass:Psi.Lip} The interaction function $\Psi$ is Lipschitz continuous in the first two variables, i.e. there exists $L_\Psi=L_\Psi(\omega)>0$ such that
$$ \abs{\Psi(v_2,v_2^\ast,\omega)-\Psi(v_1,v_1^\ast,\omega)}\leq L_\Psi(\omega)\left(\abs{v_2-v_1}+\abs{v_2^\ast-v_1^\ast}\right) $$
for all $(v_1,\,v_1^\ast),\,(v_2,\,v_2^\ast)\in\R^2$, with moreover $\ave{L_\Psi}<+\infty$.
\end{enumerate}

We begin by proving an analogue of Lemma~\ref{lemma:m1N} for a general $\Psi$ complying with this assumption.
\begin{lemma} \label{lemma:m1N.Lip}
Let
$$ \cm_{1,N}(x,t):=\int_\R\abs{v}f_N(x,v,t)\,dv, $$
$f_N$ being a solution to~\eqref{eq:Boltz-type.WN} with interaction rules~\eqref{eq:int_rules}, and assume that the interaction function $\Psi$ in~\eqref{eq:int_rules} satisfies~\ref{ass:Psi.Lip}. Then
$$ \norm{\cm_{1,N}(t)}{\infty}\leq\norm{\cm_{1,N}^0}{\infty}e^{2\ave{L_\Psi}t}
    +\frac{\ave{\abs{\Psi(0,0,\omega)}}}{2\ave{L_\Psi}}\left(e^{2\ave{L_\Psi}t}-1\right), \qquad \forall\,t>0, $$
where $\cm_{1,N}^0(x)=\cm_{1,N}(x,0)$.

The same estimate holds with $\cm_{1,N}$, $\cm_{1,N}^0$ replaced by $\cm_1$, $\cm_1^0$, where
$$ \cm_1(x,t):=\int_\R\abs{v}f(x,v,t)\,dv, \qquad \cm_1^0(x):=\cm_1(x,0) $$
and $f$ is a solution to~\eqref{eq:Boltz-type.W}.
\end{lemma}
\begin{proof}
Because of~\ref{ass:Psi.Lip}, the function $\Psi$ satisfies
$$ \abs{\Psi(v,v_\ast,\omega)}\leq\abs{\Psi(0,0,\omega)}+L_\Psi(\omega)(\abs{v}+\abs{v_\ast}). $$
Therefore, letting $\varphi(v)=\abs{v}$ in~\eqref{eq:Boltz-type.WN}, we find:
\begin{align*}
    \frac{\partial\cm_{1,N}}{\partial t}(x,t) &= \int_0^1\int_\R\int_\R W_N(x,x_\ast)\ave{\abs{\Psi(v,v_\ast,\omega)}-\abs{v}}
        f_N(x,v,t)f_N(x_\ast,v_\ast,t)\,dv\,dv_\ast\,dx_\ast \\
    &\leq \ave{\abs{\Psi(0,0,\omega)}}\cd_{W_N}(x)+(\ave{L_\Psi}-1)\cd_{W_N}(x)\cm_{1,N}(x,t) \\
    &\phantom{\leq} +\ave{L_\Psi}\int_0^1W_N(x,x_\ast)\cm_{1,N}(x_\ast,t)\,dx_\ast.
\end{align*}
Consequently,
$$ \frac{\partial\cm_{1,N}}{\partial t}(x,t)+\cd_{W_N}(x)\cm_{1,N}(x,t)\leq
    \ave{\abs{\Psi(0,0,\omega)}}\norm{\cd_{W_N}}{\infty}+2\ave{L_\Psi}\norm{\cd_{W_N}}{\infty}\norm{\cm_{1,N}(t)}{\infty}, $$
whence, multiplying both sides by $e^{\cd_{W_N}(x)t}$ and integrating over $[0,\,t]$, with $t>0$,
\begin{align*}
    e^{\cd_{W_N}(x)t}\cm_{1,N}(x,t) &\leq 
        \cm_{1,N}^0(x)+\ave{\abs{\Psi(0,0,\omega)}}\norm{\cd_{W_N}}{\infty}\int_0^te^{\cd_{W_N}(x)\tau}\,d\tau \\
    &\phantom{=} +2\ave{L_\Psi}\norm{\cd_{W_N}}{\infty}\int_0^te^{\cd_{W_N}(x)\tau}\norm{\cm_{1,N}(\tau)}{\infty}\,d\tau.
\end{align*}
Repeating the argument for $e^{\cd_{W_N}(x)t}$ used in the proof of Lemma~\ref{lemma:m1N}, we further deduce:
\begin{align*}
    e^{\dWNinf t}\norm{\cm_{1,N}(t)}{\infty} &\leq 
        \norm{\cm_{1,N}^0}{\infty}+\ave{\abs{\Psi(0,0,\omega)}}\norm{\cd_{W_N}}{\infty}\int_0^te^{\dWNinf\tau}\,d\tau \\
    &\phantom{=} +2\ave{L_\Psi}\norm{\cd_{W_N}}{\infty}\int_0^te^{\dWNinf\tau}\norm{\cm_{1,N}(\tau)}{\infty}\,d\tau,
\end{align*}
which, invoking Gr\"{o}nwall's inequality in the form provided in~\cite[Lemma~6.2]{loy2026RMUP} for the function $t\mapsto e^{\dWNinf t}\norm{\cm_{1,N}(t)}{\infty}$, gives
\begin{align*}
    \norm{\cm_{1,N}(t)}{\infty} &\leq \norm{\cm_{1,N}^0}{\infty}e^{\left(2\ave{L_\Psi}\norm{\cd_{W_N}}{\infty}-\dWNinf\right)t} \\
    &\phantom{\leq} +\ave{\abs{\Psi(0,0,\omega)}}\norm{\cd_{W_N}}{\infty}
        \int_0^te^{\left(2\ave{L_\Psi}\norm{\cd_{W_N}}{\infty}-\dWNinf\right)(t-\tau)}\,d\tau.
\end{align*}
The thesis then follows from $\dWNinf\geq 0$ and $\norm{\cd_{W_N}}{\infty}\leq 1$.

Letting $\varphi(v)=\abs{v}$ in~\eqref{eq:Boltz-type.W} and repeating the same arguments as above, with $W_N$ and $\cd_{W_N}$ replaced by $W$ and $\cd_W$, respectively, yields the conclusion also for $\cm_1$.
\end{proof}

A straightforward consequence of Lemma~\ref{lemma:m1N.Lip} is that, also in this case, $f_N(x,\cdot,t),\,f(x,\cdot,t)\in\cP_1(\R)$ for every $x\in [0,\,1]$ and every $t>0$ if the same holds at the initial time. Consequently, the distance $\cW_1(f_N(x,t),f(x,t))$ is well defined and we can prove the main result of this section:
\begin{theorem}[Dense graph limit with Lipschitz continuous interactions] \label{theo:dense_graph.Lip}
Under the same assumptions as in Theorem~\ref{theo:dense_graph.linear}, except that assumption~\ref{ass:conv.f_N^0->f^0.d1} is replaced by
\begin{enumerate}[label=(\roman*$'$)]
\setcounter{enumi}{1}
\item \label{ass:conv.f_N^0->f^0.W1} $\{f_N^0(x,\cdot)\}_{N\in\N}$ converges, for $N\to\infty$, to some $f^0=f^0(x,\cdot)\in\cP_1(\R)$ in the $1$-Wasserstein distance, i.e.
$$ \lim_{N\to\infty}\cW_1(f_N^0(x),f^0(x))=0, \qquad \text{for a.e. } x\in [0,\,1], $$
\end{enumerate}
the solutions of~\eqref{eq:Boltz-type.WN}, with interaction rules~\eqref{eq:int_rules} and interaction function $\Psi$ satisfying~\ref{ass:Psi.Lip}, converge in the $1$-Wasserstein distance to the solutions of~\eqref{eq:Boltz-type.W} emanating from $f^0$.
\end{theorem}
\begin{proof}
In~\eqref{eq:Boltz-type.WN} and~\eqref{eq:Boltz-type.W} we consider $\varphi\in\Lip_1(\R)$ with $\varphi(0)=0$ and observe, for future reference, that this implies $\abs{\varphi(v)}\leq\abs{v}$ for every $v\in\R$.

Subtracting~\eqref{eq:Boltz-type.WN} from~\eqref{eq:Boltz-type.W} and rearranging the terms, we get:
\begin{align}
    \begin{aligned}[b]
    \frac{\partial}{\partial t} &\int_\R\varphi(v)(f(x,v,t)-f_N(x,v,t))\,dv \\
    &= \int_0^1\int_\R\int_\R W(x,x_\ast)\ave{\varphi(v')}(f(x,v,t)-f_N(x,v,t))f(x_\ast,v_\ast,t)\,dv\,dv_\ast\,dx_\ast \\
    &\phantom{=} +\int_0^1\int_\R\int_\R W_N(x,x_\ast)\ave{\varphi(v')}f_N(x,v,t)(f(x_\ast,v_\ast,t)-f_N(x_\ast,v_\ast,t))\,dv\,dv_\ast\,dx_\ast \\
    &\phantom{=} +\int_0^1(W(x,x_\ast)-W_N(x,x_\ast))\left(\int_\R\int_\R\ave{\varphi(v')}f_N(x,v,t)f(x_\ast,v_\ast,t)\,dv\,dv_\ast\right. \\
    &\phantom{=+\int_0^1(W(x,x_\ast)-W_N(x,x_\ast))} \left.-\int_\R\varphi(v)f(x,v,t)\,dv\right)dx_\ast \\
    &\phantom{=} -\cd_{W_N}(x)\int_\R\varphi(v)(f(x,v,t)-f_N(x,v,t))\,dv.
    \end{aligned}
    \label{eq:W1.rearrangement-proof}
\end{align}

We now analyse the first three terms on the right-hand side.

First, we have
\begin{align*}
    \int_0^1\int_\R\int_\R &W(x,x_\ast)\ave{\varphi(v')}(f(x,v,t)-f_N(x,v,t))f(x_\ast,v_\ast,t)\,dv\,dv_\ast\,dx_\ast \\
    &=\int_0^1\int_\R W(x,x_\ast)f(x_\ast,v_\ast,t)\left(\int_\R\ave{\varphi(v')}(f(x,v,t)-f_N(x,v,t))\,dv\right)dv_\ast\,dx_\ast
\intertext{and we observe that, for every fixed $v_\ast\in\R$ and $\omega\in\Omega$, the mapping $v\mapsto\ave{\varphi(v')}=\ave{\varphi(\Psi(v,v_\ast,\omega))}$ is Lipschitz continuous by composition, because so are both $\varphi$ and $\Psi(\cdot,v_\ast,\omega)$. Thus,}
    &\leq \ave{L_\Psi}\int_0^1\int_\R W(x,x_\ast)f(x_\ast,v_\ast,t)\cW_1(f_N(x,t),f(x,t))\,dv_\ast\,dx_\ast \\
    &= \ave{L_\Psi}\cd_W(x)\cW_1(f_N(x,t),f(x,t)).
\end{align*}

Analogously,
\begin{align*}
    \int_0^1\int_\R\int_\R &W_N(x,x_\ast)\ave{\varphi(v')}f_N(x,v,t)(f(x_\ast,v_\ast,t)-f_N(x_\ast,v_\ast,t))\,dv\,dv_\ast\,dx_\ast \\
    &= \int_0^1\int_\R W_N(x,x_\ast)f_N(x,v,t)\left(\int_\R\ave{\varphi(v')}(f(x_\ast,v_\ast,t)-f_N(x_\ast,v_\ast,t))\,dv_\ast\right)dv\,dx_\ast \\
    &\leq \ave{L_\Psi}\int_0^1\int_\R W_N(x,x_\ast)f_N(x,v,t)\cW_1(f_N(x_\ast,t),f(x_\ast,t))\,dv\,dx_\ast \\
    &= \ave{L_\Psi}\int_0^1 W_N(x,x_\ast)\cW_1(f_N(x_\ast,t),f(x_\ast,t))\,dx_\ast,
\end{align*}
where now we have used the fact that, for every fixed $v\in\R$ and $\omega\in\Omega$, the mapping $v_\ast\mapsto\ave{\varphi(v')}$ is Lipschitz continuous.

Finally,
\begin{align*}
    &\int_0^1(W(x,x_\ast)-W_N(x,x_\ast))\left(\int_\R\int_\R\ave{\varphi(v')}f_N(x,v,t)f(x_\ast,v_\ast,t)\,dv\,dv_\ast\right. \\
    &\phantom{\int_0^1(W(x,x_\ast)-W_N(x,x_\ast))} \left.-\int_\R\varphi(v)f(x,v,t)\,dv\right)dx_\ast \\
    &= \int_0^1(W(x,x_\ast)-W_N(x,x_\ast))\left(\int_\R\int_\R\ave{\varphi(v')-\varphi(v)}f_N(x,v,t)f(x_\ast,v_\ast,t)\,dv\,dv_\ast\right. \\
    &\phantom{\int_0^1(W(x,x_\ast)-W_N(x,x_\ast))} \left.-\int_\R\varphi(v)(f(x,v,t)-f_N(x,v,t))\,dv\right)dx_\ast.
\end{align*}
We observe, in particular, that
\begin{multline*}
    \abs*{\int_\R\int_\R\ave{\varphi(v')-\varphi(v)}f_N(x,v,t)f(x_\ast,v_\ast,t)\,dv\,dv_\ast-\int_\R\varphi(v)(f(x,v,t)-f_N(x,v,t))\,dv} \\
        \leq\int_\R\int_\R\ave{\abs{\varphi(v')-\varphi(v)}}f_N(x,v,t)f(x_\ast,v_\ast,t)\,dv\,dv_\ast
            +\int_\R\abs{v}(f(x,v,t)+f_N(x,v,t))\,dv
\end{multline*}
and that $\ave{\abs{\varphi(v')-\varphi(v)}}\leq\ave{\abs{\Psi(v,v_\ast,\omega)-v}}\leq\ave{\abs{\Psi(0,0,\omega)}} (\ave{L_\Psi}+1)\abs{v}+\ave{L_\Psi}\abs{v_\ast}$, whence
\begin{multline*}
    \abs*{\int_\R\int_\R\ave{\varphi(v')-\varphi(v)}f_N(x,v,t)f(x_\ast,v_\ast,t)\,dv\,dv_\ast-\int_\R\varphi(v)(f(x,v,t)-f_N(x,v,t))\,dv} \\
    \leq\ave{\abs{\Psi(0,0,\omega)}}+(\ave{L_\Psi}+2)\norm{\cm_{1,N}(t)}{\infty}+(\ave{L_\Psi}+1)\norm{\cm_1(t)}{\infty}.
\end{multline*}
Invoking Lemma~\ref{lemma:m1N.Lip}, along with the uniform boundedness of the $\cm_{1,N}^0$'s with respect to $N$, we deduce that there exists an exponentially growing function $\alpha=\alpha(t)$, independent of $N$, such that
\begin{align*}
    &\int_0^1(W(x,x_\ast)-W_N(x,x_\ast))\left(\int_\R\int_\R\ave{\varphi(v')}f_N(x,v,t)f(x_\ast,v_\ast,t)\,dv\,dv_\ast\right. \\
    &\phantom{\int_0^1(W(x,x_\ast)-W_N(x,x_\ast))} \left.-\int_\R\varphi(v)f(x,v,t)\,dv\right)dx_\ast\leq
        \alpha(t)\int_0^1\abs{W(x,x_\ast)-W_N(x,x_\ast)}\,dx_\ast.
\end{align*}

On the whole, we continue~\eqref{eq:W1.rearrangement-proof} as
\begin{align*}
    \frac{\partial}{\partial t} &\left(e^{\cd_{W_N}(x)t}\int_\R\varphi(v)(f(x,v,t)-f_N(x,v,t))\,dv\right) \\
    &\leq e^{\cd_{W_N}(x)t}\ave{L_\Psi}\left(\norm{\cd_W}{\infty}\cW_1(f_N(x,t),f(x,t))
        +\int_0^1W_N(x,x_\ast)\cW_1(f_N(x_\ast,t),f(x_\ast,t))\,dx_\ast\right) \\
    &\phantom{\leq} +\alpha(t)e^{\cd_{W_N}(x)t}\int_0^1\abs{W(x,x_\ast)-W_N(x,x_\ast)}\,dx_\ast,
\end{align*}
whence, integrating both sides over $[0,\,t]$, for $t>0$, and taking the supremum over $\varphi$ on the left-hand side,
\begin{align*}
    e^{\cd_{W_N}(x)t}\cW_1(f_N(x,t),f(x,t)) &\leq \cW_1(f_N^0(x),f^0(x)) \\
    &\phantom{\leq} +\ave{L_\Psi}\norm{\cd_W}{\infty}\int_0^te^{\cd_{W_N}(x)\tau}\cW_1(f_N(x,\tau),f(x,\tau))\,d\tau \\
    &\phantom{\leq} +\ave{L_\Psi}\int_0^te^{\cd_{W_N}(x)\tau}\int_0^1W_N(x,x_\ast)\cW_1(f_N(x_\ast,t),f(x_\ast,t))\,dx_\ast\,d\tau \\
    &\phantom{\leq} +\psi(t)\int_0^1\abs{W(x,x_\ast)-W_N(x,x_\ast)}\,dx_\ast,
\end{align*}
where we have defined
$$ \psi(t):=\int_0^t\alpha(\tau)e^{\dWNinf\tau}\,d\tau. $$
Repeating the argument for $e^{\cd_{W_N}(x)t}$ used in the proof of Lemma~\ref{lemma:m1N} and integrating then both sides over $[0,\,1]$ with respect to $x$, we discover
\begin{align*}
    e^{\dWNinf t} &\int_0^1\cW_1(f_N(x,t),f(x,t))\,dx \\ 
    &\leq \int_0^1\cW_1(f_N^0(x),f^0(x))\,dx+\psi(t)\norm{W-W_N}{L^1([0,\,1]^2)} \\
    &\phantom{\leq} +\ave{L_\Psi}(\norm{\cd_W}{\infty}+\norm{\cd_{W_N}}{\infty})\int_0^te^{\dWNinf\tau}\int_0^1\cW_1(f_N(x,\tau),f(x,\tau))\,dx\,d\tau.
\end{align*}
At this point, Gr\"{o}nwall's inequality applied to the function $t\mapsto e^{\dWNinf t}\int_0^1\cW_1(f_N(x,t),f(x,t))\,dx$ implies
\begin{multline*}
    \int_0^1\cW_1(f_N(x,t),f(x,t))\,dx \\
        \leq\left(\int_0^1\cW_1(f_N^0(x),f^0(x))\,dx+\psi(t)\norm{W-W_N}{L^1([0,\,1]^2)}\right)
            e^{\left(\ave{L_\Psi}(\norm{\cd_W}{\infty}+\norm{\cd_{W_N}}{\infty})-\dWNinf\right)t},
\end{multline*}
whence the thesis follows by an argument analogous to that used in the proof of Theorem~\ref{theo:dense_graph.linear}.
\end{proof}

\subsubsection{Interaction rules with confined trait~\texorpdfstring{$\boldsymbol{v}$}{}}
In several applications, the trait $v$ of the agents does not span the whole real line, but instead belongs to a subset $K\subseteq\R$, which may be either unbounded, e.g. $K=\R_+$ in wealth redistribution~\cite{cordier2005JSP} or social contacts~\cite{dimarco2021JMB,loy2021MBE} problems, or bounded, e.g. $K=[-1,\,1]$ in opinion formation problems~\cite{toscani2006CMS} or $K=[0,\,1]$ in epidemiological problems~\cite{dellamarca2022NHM,dellamarca2023JMB}, to mention just a few examples. In these cases, homogeneous Boltzmann-type equations are formulated by integrating over $K$, rather than over $\R$, with respect to $v$ and $v_\ast$.

If such problems are posed on networks, the question arises as to whether~\eqref{eq:Boltz-type.W} can still be regarded as the dense graph limit of~\eqref{eq:Boltz-type.WN} when the agents' trait is confined to $K$.

Assume that the interaction function $K$ satisfies the following property:
\begin{enumerate}[label=(H\arabic*)]
\setcounter{enumi}{1}
\item \label{ass:Psi.K} For every $v,\,v_\ast\in K$ and $\omega\in\Omega$, it results $\Psi(v,v_\ast,\omega)\in K$.
\end{enumerate}
Then we have:
\begin{theorem} \label{theo:suppf.K}
Under assumption~\ref{ass:Psi.K}, solutions of~\eqref{eq:Boltz-type.W} are such that if $\supp{f^0(x,\cdot)}\subseteq K$ for every $x\in [0,\,1]$ then $\supp{f(x,\cdot,t)}\subseteq K$ for every $x\in [0,\,1]$ and every $t>0$.

The same holds for solutions of~\eqref{eq:Boltz-type.WN}, with $f^0$, $f$ replaced by $f_N^0$, $f_N$, respectively.
\end{theorem}
\begin{proof}
We prove the statement working explicitly on~\eqref{eq:Boltz-type.W}.

Since $f^0(x,\cdot)$ is a probability measure, the condition $\supp{f^0(x,\cdot)}\subseteq K$ for every $x\in [0,\,1]$ is equivalent to
$$ \int_Kf^0(x,v)\,dv=1, \qquad \forall\,x\in [0,\,1]. $$
At the same time, since $f(x,\cdot,t)$ is a probability measure on $\R$, we have
$$ \int_Kf(x,v,t)\,dv\leq 1, \qquad \forall\,x\in [0,\,1],\ \forall\,t>0. $$
We will show that this inequality is, in fact, an equality.

In~\eqref{eq:Boltz-type.W}, let $\varphi(v)=\chi_K(v)$, the latter being the characteristic function of the set $K$. We get:
\begin{align*}
    \frac{\partial}{\partial t}\int_Kf(x,v,t)\,dv &= \int_0^1\int_\R\int_\R W(x,x_\ast)\ave{\chi_K(v')}f(x,v,t)f(x_\ast,v_\ast,t)\,dv\,dv_\ast\,dx_\ast \\
    &\phantom{=} -\cd_W(x)\int_Kf(x,v,t)\,dv\,dx \\
    &\geq \int_0^1\int_K\int_KW(x,x_\ast)f(x,v,t)f(x_\ast,v_\ast,t)\,dv\,dv_\ast\,dx_\ast \\
    &\phantom{=} -\cd_W(x)\int_Kf(x,v,t)\,dv,
\intertext{where we have used that $\chi_K(v')=\chi_K(\Psi(v,v_\ast,\omega))=1$ for $v,\,v_\ast\in K$ because of~\ref{ass:Psi.K}. Furthermore,}
    &= \int_0^1W(x,x_\ast)\left(\int_Kf(x,v,t)\,dv\right)\left(\int_Kf(x_\ast,v_\ast,t)\,dv_\ast\right)\,dx_\ast \\
    &\phantom{=} -\cd_W(x)\int_Kf(x,v,t)\,dv.
\end{align*}

Introducing, for brevity, the function
$$ h(x,t):=\int_Kf(x,v,t)\,dv, $$
which is such that $0\leq h\leq 1$, we may rewrite the previous relation as
\begin{align*}
    \frac{\partial h}{\partial t}(x,t) &\geq h(x,t)\int_0^1W(x,x_\ast)h(x_\ast,t)\,dx_\ast-\cd_W(x)h(x,t) \\
    &\geq \cd_W(x)h(x,t)\left(\inf_{x_\ast\in [0,\,1]}h(x_\ast,t)-1\right) \\
    &\geq \norm{\cd_W}{\infty}h(x,t)\left(\inf_{x_\ast\in [0,\,1]}h(x_\ast,t)-1\right),
\end{align*}
the last inequality following from the fact that $\inf_{x_\ast\in [0,\,1]}h(x_\ast,t)-1\leq 0$. After a further manipulation, we arrive at
$$ \frac{\partial h}{\partial t}(x,t)+\norm{\cd_W}{\infty}h(x,t)\geq\norm{\cd_W}{\infty}\left(\inf_{x\in [0,\,1]}h(x,t)\right)^2, $$
whence, multiplying both sides by $e^{\norm{\cd_W}{\infty}t}$ and integrating over $[0,\,t]$, with $t>0$,
\begin{align*}
    e^{\norm{\cd_W}{\infty}t}h(x,t) &\geq
        h(x,0)+\norm{\cd_W}{\infty}\int_0^te^{\norm{\cd_W}{\infty}\tau}\left(\inf_{x\in [0,\,1]}h(x,\tau)\right)^2\,d\tau \\
    &= 1+\norm{\cd_W}{\infty}\int_0^te^{\norm{\cd_W}{\infty}\tau}\left(\inf_{x\in [0,\,1]}h(x,\tau)\right)^2\,d\tau
\end{align*}
as $h(x,0)=\int_Kf^0(x,v)\,dv=1$ by assumption. Since the right-hand side is independent of $x$, the same inequality remains valid upon taking the infimum over $x$ on the left-hand side. Altogether, defining $u(t):=e^{\norm{\cd_W}{\infty}t}\inf_{x\in [0,\,1]}h(x,t)$, we obtain:
$$ u(t)\geq 1+\norm{\cd_W}{\infty}\int_0^te^{-\norm{\cd_W}{\infty}\tau}u^2(\tau)\,d\tau, $$
whence the (reversed) Gr\"{o}nwall-Bihari inequality, cf.~\cite{bihari1956AMASH,lasalle1949AM}, yields $u(t)\geq e^{\norm{\cd_W}{\infty}t}$, and therefore $\inf_{x\in [0,\,1]}h(x,t)\geq 1$, for $t>0$. We conclude that $h(x,t)\geq 1$, and consequently that $h(x,t)=1$, for every $x\in [0,\,1]$ and $t>0$, which establishes the claim.
\end{proof}

Theorem~\ref{theo:suppf.K} implies that the homogeneous Boltzmann-type equations~\eqref{eq:Boltz-type.WN} and~\eqref{eq:Boltz-type.W} remain valid in the above form, namely with integrals over $\R$ rather than over $K$, even when the trait $v$ is confined to $K$, since, under assumption~\ref{ass:Psi.K}, the distribution functions $f_N$, $f$ are automatically supported in $K$. Consequently, the conclusion of Theorem~\ref{theo:dense_graph.Lip} continues to hold unchanged for $v\in K$.

\begin{remark}
A noteworthy case in which the conclusions above apply is for linear interaction rules~\eqref{eq:int_rules},~\eqref{eq:lin_int} with $p,\,q\geq 0$ and $K=\R_+$. This setting is typical, for example, of models for wealth distribution and for the transmission of infectious diseases, see e.g.,~\cite{bisi2017BUMI,bisi2009CMS,dimarco2020PRE,lanchier2018JSP,loy2021KRM,loy2021MBE}.
\end{remark}

\section{On the relaxation to equilibrium}
\label{sect:relax_equil}
In the kinetic theory of multi-agent systems, equilibrium solutions to Boltzmann-type equations describe the self-organised configurations emerging spontaneously from mutual pairwise interactions among the agents. Conceptually, they are the analogue of the \textit{Maxwellian distribution} in the classical kinetic theory of gases.

In this section, we address the relaxation to equilibrium of solutions to the homogeneous Boltzmann-type equation~\eqref{eq:Boltz-type.W} on a dense graph derived in Section~\ref{sect:kin_eq.graph}. To be definite, we consider the class of linear symmetric interaction rules~\eqref{eq:int_rules},~\eqref{eq:lin_int}, which allow for quite detailed analytical investigations of the resulting kinetic equation.

The relaxation to equilibrium of the solutions to homogeneous Boltzmann-type equations with linear interactions may be conveniently approached by means of the Fourier distances, which in several cases are contractive in time on the solutions of those equations under suitable assumptions on the interaction coefficients $p,\,q$, cf.~\cite{carrillo2007RMUP}.

In parallel to the family of spaces $\cP_s(\R)$ introduced in Section~\ref{sect:lin_int}, for subsequent purposes it is useful to introduce also the set $\cP_{s,\alpha,\cM_{s+\alpha}}(\R)$ with $\alpha>0$. This is defined as the subset of $\cP_{s+\alpha}(\R)$ made of probability measures $\mu$ which have prescribed moments up to the order $s$ and are such that $\int_\R\abs{v}^{s+\alpha}\,d\mu(v)$ is uniformly bounded with respect to $\mu$ by a constant $\cM_{s+\alpha}\geq 0$. The reason why such a subset is relevant is expressed by the following result, see again~\cite{carrillo2007RMUP}:
\begin{lemma} \label{lemma:P.complete_metric_space}
The set $\cP_{s,\alpha,\cM_{s+\alpha}}(\R)$, endowed with the distance $d_s$, is a complete metric space.
\end{lemma}

Equation~\eqref{eq:Boltz-type.W.Fourier}, namely the Fourier-transformed form of~\eqref{eq:Boltz-type.W}, will serve as the basis for the subsequent analysis.

\subsection{Long-time asymptotics}
\label{sect:long-time_asymptotics}
We begin by addressing the long-time behaviour of the solutions to~\eqref{eq:Boltz-type.W}. In particular, for a certain $s\in\N$, $s>0$, to be subsequently discussed, we assume that~\eqref{eq:Boltz-type.W} admits solutions in $\cP_s(\R)$ and we let $f(x,\cdot,t),\,g(x,\cdot,t)\in\cP_s(\R)$ be any two of them for which the distance $d_s$ is well defined, i.e. finite. According to Lemma~\ref{lemma:ds.finite}, this means that $f(x,\cdot,t),\,g(x,\cdot,t)$ have the same moments at least up to the order $s-1$ for every $x\in [0,\,1]$ and every $t>0$. We denote by $f^0(x,\cdot),\,g^0(x,\cdot)\in\cP_s(\R)$ the respective initial conditions.

We begin by establishing the following \textit{a priori} estimate:
\begin{theorem} \label{theo:ds}
In the said assumptions, it results
$$ \sup_{x\in [0,\,1]}d_s(f(x,t),g(x,t))\leq\sup_{x\in [0,\,1]}d_s(f^0(x),g^0(x))e^{\left(\ave{p^s+q^s}\norm{\cd_W}{\infty}-\dWinf\right)t},
    \qquad \forall\,t>0. $$
\end{theorem}
\begin{proof}
The proof proceeds largely along the same lines as that of Theorem~\ref{theo:dense_graph.linear}, formally with $W_N$ replaced by $W$ and $f_N$ by $g$. However, instead of dividing by $\abs{\xi}$, we now divide by $\abs{\xi}^s$, so that~\eqref{eq:dense_graph_limit.proof} becomes:
\begin{multline*}
    \frac{\partial}{\partial t}\left(e^{\cd_W(x)t}\frac{\abs{\hat{g}(x,\xi,t)-\hat{f}(x,\xi,t)}}{\abs{\xi}^s}\right) \\
    \leq e^{\cd_W(x)t}\left[\norm{\cd_W}{\infty}\ave{p^s}d_s(f(x,t),g(x,t))+\ave{q^s}\int_0^1W(x,x_\ast)d_s(f(x_\ast,t),g(x_\ast,t))\,dx_\ast\right].
\end{multline*}
We now integrate both sides over $[0,\,t]$, with $t>0$, and treat $e^{\cd_W(x)t}$ analogously to the proof of Lemma~\ref{lemma:m1N}:
\begin{align*}
    e^{\dWinf t}\frac{\abs{\hat{g}(x,\xi,t)-\hat{f}(x,\xi,t)}}{\abs{\xi}^s} &\leq d_s(f^0(x),g^0(x)) \\
    &\phantom{\leq} +\int_0^te^{\dWinf\tau}\biggl[\norm{\cd_W}{\infty}\ave{p^s}d_s(f(x,\tau),g(x,\tau)) \\
    &\phantom{\leq} \qquad+\ave{q^s}\int_0^1W(x,x_\ast)d_s(f(x_\ast,\tau),g(x_\ast,\tau))\,dx_\ast\biggr]d\tau \\
    &\leq \sup_{x\in [0,\,1]}d_s(f^0(x),g^0(x)) \\
    &\phantom{\leq} +\ave{p^s+q^s}\norm{\cd_W}{\infty}\int_0^te^{\dWinf\tau}\sup_{x\in [0,\,1]}d_s(f(x,\tau),g(x,\tau))\,d\tau.
\end{align*}
Taking the supremum over $\xi\in\R\setminus\{0\}$ and then over $x\in [0,\,1]$ on the left hand-side produces
\begin{multline*}
    e^{\dWinf t}\sup_{x\in [0,\,1]}d_s(f(x,t),g(x,t))\leq\sup_{x\in [0,\,1]}d_s(f^0(x),g^0(x)) \\
    +\ave{p^s+q^s}\norm{\cd_W}{\infty}\int_0^te^{\dWinf\tau}\sup_{x\in [0,\,1]}d_s(f(x,\tau),g(x,\tau))\,d\tau,
\end{multline*}
whence, applying Gr\"{o}nwall's inequality to the function $t\mapsto e^{\dWinf t}\sup_{x\in [0,\,1]}d_s(f(x,t),g(x,t))$, we obtain the thesis.
\end{proof}

We observe that Theorem~\ref{theo:ds} implies, as a by-product, uniqueness of the solution to~\eqref{eq:Boltz-type.W}. Moreover, it enables us to characterise precisely the asymptotic behaviour of the solutions under suitable assumptions on the coefficients of the linear interaction rules~\eqref{eq:int_rules},~\eqref{eq:lin_int}. More specifically, we have the following:
\begin{theorem} \label{theo:M}
In the same assumptions as those of Theorem~\ref{theo:ds}, if $\dWinf>0$ and
\begin{equation}
    \ave{p^s+q^s}<\frac{\dWinf}{\norm{\cd_W}{\infty}}
    \label{eq:ps+qs}
\end{equation}
there exists at most one equilibrium distribution $M=M(x,v)$ of~\eqref{eq:Boltz-type.W}, with $M(x,\cdot)\in\cP_s(\R)$ for every $x\in [0,\,1]$, which is globally asymptotically stable.
\end{theorem}
\begin{proof}
Equilibrium distributions are constant-in-time solutions of~\eqref{eq:Boltz-type.W}. If, for every $x\in [0,\,1]$, we assume that $M_1(x,\cdot),\,M_2(x,\cdot)\in\cP_s(\R)$ are any two of them, from Theorem~\ref{theo:ds} we deduce
$$ \sup_{x\in [0,\,1]}d_s(M_1(x),M_2(x))\leq \sup_{x\in [0,\,1]}d_s(M_1(x),M_2(x))e^{\left(\ave{p^s+q^s}\norm{\cd_W}{\infty}-\dWinf\right)t},
    \qquad \forall\,t>0. $$
But under~\eqref{eq:ps+qs} we have $e^{\left(\ave{p^s+q^s}\norm{\cd_W}{\infty}-\dWinf\right)t}<1$ for $t>0$, therefore
$$ \sup_{x\in [0,\,1]}d_s(M_1(x),M_2(x))=0 $$
and the uniqueness follows.

Moreover, still owing to~\eqref{eq:ps+qs}, for any solution $f(x,\cdot,t)\in\cP_s(\R)$ to~\eqref{eq:Boltz-type.W} Theorem~\ref{theo:ds} implies
$$ \lim_{t\to +\infty}\sup_{x\in [0,\,1]}d_s(f(x,t),M(x))=0 $$
independently of the initial condition $f^0$, whence the global stability and attractiveness of $M$.
\end{proof}

We now discuss the choice of the index $s$ mentioned at the beginning of this section.

\subsubsection{The case~\texorpdfstring{$\boldsymbol{s=1}$}{}}
Since every solution $f$ of~\eqref{eq:Boltz-type.W} is such that $\int_\R f(x,v,t)\,dv=1$ for every $x\in [0,\,1]$ and every $t>0$, cf. the beginning of the proof of Theorem~\ref{theo:dense_graph.linear}, according to Lemma~\ref{lemma:ds.finite} the Fourier distance $d_1$ is always well defined for any pair of solutions of~\eqref{eq:Boltz-type.W}. Moreover, from Lemma~\ref{lemma:m1N} applied formally to $f$ we know that
\begin{equation}
    \norm{\cm_1(t)}{\infty}\leq\norm{\cm_1^0}{\infty}e^{\left(\ave{p+q}\norm{\cd_W}{\infty}-\dWinf\right)t},
    \label{eq:m1}
\end{equation}
which shows that $\cm_1$ is bounded for every $x\in [0,\,1]$ and every $t>0$. Therefore, we conclude that $f(x,\cdot,t)\in\cP_1(\R)$ and that Theorems~\ref{theo:ds},~\ref{theo:M} apply with $s=1$. Specifically, if
$$ \ave{p+q}<\frac{\dWinf}{\norm{\cd_W}{\infty}} $$
there exists at most one equilibrium distribution $M(x,\cdot)\in\cP_1(\R)$, which is globally asymptotically stable. In addition to this,~\eqref{eq:m1} shows that, in this case, $\cm_1(x,t)\to 0$ uniformly in $x$ as $t\to +\infty$. Since the first moment
$$ m_1(x,t):=\int_\R vf(x,v,t)\,dv $$
is such that $\abs{m_1(x,t)}\leq\cm_1(x,t)$ for all $(x,\,t)\in [0,\,1]\times [0,\,+\infty)$, the same is true also for $m_1$, which enables us to conclude that the prospective equilibrium solution $M$ satisfies
$$ \int_\R vM(x,v)\,dv=0, \qquad \forall\,x\in [0,\,1]. $$

\subsubsection{The case~\texorpdfstring{$\boldsymbol{s=2}$}{}}
Substituting $\varphi(v)=v$ into~\eqref{eq:Boltz-type.W}, we observe that the first moment $m_1$ introduced above satisfies the integro-differential equation
$$ \frac{\partial m_1}{\partial t}(x,t)=(\ave{p}-1)\cd_W(x)m_1(x,t)+\ave{q}\int_0^1W(x,x_\ast)m_1(x_\ast,t)\,dx_\ast. $$
By linearity, it is not difficult to see that this equation admits at most one solution for every prescribed initial condition $m_1^0\in L^\infty(0,\,1)$. In particular, if $m_1,\,n_1$ are two solutions issuing from the initial conditions $m_1^0,\,n_1^0$, respectively, the following continuous dependence estimate holds true:
$$ \norm{n_1(t)-m_1(t)}{\infty}\leq \norm{n_1^0-m_1^0}{\infty}e^{\left(\ave{p+q}\norm{\cd_W}{\infty}-\dWinf\right)t},
    \qquad\forall\,t>0. $$
By direct substitution into the previous equation, we discover that if $\ave{p+q}=1$ and the initial condition is constant in $x$, say $m_1^0(x)=\bar{m}_1^0\in\R$ for every $x\in [0,\,1]$, the unique solution is $m_1(x,t)=\bar{m}_1^0$ for every $(x,\,t)\in [0,\,1]\times (0,\,+\infty)$.

Hence, under the assumption $\ave{p+q}=1$ together with an initial condition with fixed constant first moment, the solutions of~\eqref{eq:Boltz-type.W} have equal moments up to the order $1$ for every $x\in [0,\,1]$ and every $t>0$. Therefore, owing to Lemma~\ref{lemma:ds.finite}, their Fourier distance $d_2$ is well defined.

Additionally, such solutions belong to $\cP_2(\R)$ for every $(x,\,t)\in [0,\,1]\times (0,\,+\infty)$ if the initial condition possesses a second moment bounded in $[0,\,1]$. In fact, substituting $\varphi(v)=v^2$ in~\eqref{eq:Boltz-type.W} and invoking the arbitrariness of $\phi$ we see that the second moment
$$ m_2(x,t)=\cm_2(x,t)=\int_\R v^2f(x,v,t)\,dv $$
of any solution $f$ of~\eqref{eq:Boltz-type.W} satisfies
\begin{align*}
    \frac{\partial m_2}{\partial t}(x,t) &= (\ave{p^2}-1)\cd_W(x)m_2(x,t) \\
    &\phantom{\leq} +2\ave{pq}m_1(x,t)\int_0^1W(x,x_\ast)m_1(x_\ast,t)\,dx_\ast
        +\ave{q^2}\int_0^1W(x,x_\ast)m_2(x_\ast,t)\,dx_\ast,
\end{align*}
whence, arguing along the same lines as for $\cm_1$ and making use of Gr\"{o}nwall's inequality in the form provided in~\cite[Lemma~6.2]{loy2026RMUP},
\begin{align*}
    \norm{m_2(t)}{\infty} &\leq \norm{m_2^0}{\infty}e^{\left(\ave{p^2+q^2}\norm{\cd_W}{\infty}-\dWinf\right)t} \\
    &\phantom{\leq} +2\ave{pq}\norm{\cd_W}{\infty}\int_0^t\norm{m_1(\tau)}{\infty}^2e^{\left(\ave{p^2+q^2}\norm{\cd_W}{\infty}-\dWinf\right)(t-\tau)}\,d\tau.
\end{align*}
For a constant $m_1$, this produces further
\begin{align}
    \begin{aligned}[b]
        \norm{m_2(t)}{\infty} &\leq \norm{m_2^0}{\infty}e^{\left(\ave{p^2+q^2}\norm{\cd_W}{\infty}-\dWinf\right)t} \\
        &\phantom{\leq} +\frac{2\ave{pq}\norm{\cd_W}{\infty}{(\bar{m}_1^0)}^2}{\dWinf-\ave{p^2+q^2}\norm{\cd_W}{\infty}}
            \left(1-e^{\left(\ave{p^2+q^2}\norm{\cd_W}{\infty}-\dWinf\right)t}\right),
    \end{aligned}
    \label{eq:m2}
\end{align}
which shows the uniform boundedness of $m_2$ at all times $t>0$ if $\norm{m_2^0}{\infty}<+\infty$.

Therefore, under the assumptions
$$ \ave{p+q}=1, \qquad m_1^0(x)=\bar{m}_1^0\in\R\quad\forall\,x\in [0,\,1], $$
we conclude that $f(x,\cdot,t)\in\cP_2(\R)$ for all $(x,\,t)\in [0,\,1]\times (0,\,+\infty)$ and that Theorems~\ref{theo:ds},~\ref{theo:M} apply with $s=2$. In particular, if
$$ \ave{p^2+q^2}<\frac{\dWinf}{\norm{\cd_W}{\infty}} $$
there exists at most one equilibrium distribution $M(x,\cdot)\in\cP_2(\R)$ with mean $\bar{m}_1^0$, which is globally asymptotically stable.

Unlike the case $s=1$ discussed before, now the mean of $M$ need not be zero. On the other hand, the estimate~\eqref{eq:m2} provides a precise characterisation of the second moment -- viz. the total ``energy'' -- of $M$ only if $\bar{m}_1^0=0$: in such a case, it gives $m_2(x,t)\to 0$ uniformly in $x$ when $t\to +\infty$ and, ultimately, 
$$ M(x,v)=\delta_0(v), \qquad \forall\,x\in [0,\,1], $$
where $\delta_0$ denotes the Dirac delta centred at $0$. Conversely, if $\bar{m}_1^0\neq 0$ then $M$ can be, in principle, any probability distribution with non-zero mean and energy. We will be more precise about this in the next section.

\subsection{Existence of equilibrium distributions}
The strong form of~\eqref{eq:Boltz-type.W} with linear symmetric interaction rules~\eqref{eq:int_rules},~\eqref{eq:lin_int} is, for $p\neq q$,
\begin{equation}
    \frac{\partial f}{\partial t}(x,v,t)=\int_0^1\int_\R\frac{1}{\abs{p^2-q^2}}W(x,x_\ast)\ave{f(x,\pr{v},t)f(x_\ast,\pr{v}_\ast,t)}\,dv_\ast\,dx_\ast
        -\cd_W(x)f(x,v,t),
    \label{eq:Boltz.W-strong}
\end{equation}
where $\abs{p^2-q^2}$ is the Jacobian factor of the transformation $(v,\,v_\ast)\mapsto (v',\,v_\ast')$ and $\pr{v},\,\pr{v}_\ast$ express the \textit{inverse interaction rules}:
$$ \pr{v}=\frac{p}{p^2-q^2}v-\frac{q}{p^2-q^2}v_\ast, \qquad
    \pr{v}_\ast=\frac{p}{p^2-q^2}v_\ast-\frac{q}{p^2-q^2}v, $$
i.e. those yielding the pre-interaction traits $\pr{v},\,\pr{v}_\ast$ as functions of the post-interaction ones $v,\,v_\ast$.

The right-hand side of~\eqref{eq:Boltz.W-strong} is the collision operator $Q(f,f)$. In particular, its first term is the so-called \textit{gain operator} $Q^+(f,f)$:
$$ Q^+(f,f)(x,v,t):=\int_0^1\int_\R\frac{1}{\abs{p^2-q^2}}W(x,x_\ast)\ave{f(x,\pr{v},t)f(x_\ast,\pr{v}_\ast,t)}\,dv_\ast\,dx_\ast, $$
which is such that
\begin{equation}
    \int_\R\varphi(v)Q^+(f,f)(x,v,t)\,dv=\int_0^1\int_\R\int_\R W(x,x_\ast)\ave{\varphi(v')}f(x,v,t)f(x_\ast,v_\ast,t)\,dv\,dv_\ast\,dx_\ast
    \label{eq:Q+.weak}
\end{equation}
for every $\varphi:\R\to\R$.

Equilibrium distributions $M$ are solutions of~\eqref{eq:Boltz.W-strong} independent of $t$, thus such that $\frac{\partial M}{\partial t}\equiv 0$. Therefore, they satisfy
\begin{equation}
    M(x,v)=\frac{1}{\cd_W(x)}Q^+(M,M)(x,v),
    \label{eq:fixed_point}
\end{equation}
i.e. they can be regarded as fixed points of the operator $\frac{1}{\cd_W(x)}Q^+$.

To seek such fixed points, it is convenient to formulate the problem in a complete metric space. For this, we take advantage of Lemma~\ref{lemma:P.complete_metric_space} distinguishing the cases $s=1$ and $s=2$.

\subsubsection{The case~\texorpdfstring{$\boldsymbol{s=1}$}{}}
From Section~\ref{sect:long-time_asymptotics}, we expect that for $s=1$ the prospective equilibrium solution has constant zero mean. Hence, the idea is to work in $\cP_{1,\alpha,\cM_{1+\alpha}}(\R)$ for some $\alpha>0$, fixing $m_1\equiv 0$. To choose $\alpha$, we rely on the following result:
\begin{lemma}
If $\ave{p^2+q^2}<1$, any solution to~\eqref{eq:fixed_point} with $m_1\equiv 0$ has also zero second moment.
\end{lemma}
\begin{proof}
Recalling~\eqref{eq:Q+.weak}, we multiply both sides of~\eqref{eq:fixed_point} by $v^2$ and integrate over $v$ to find
\begin{align*}
    m_2(x) &= \frac{1}{\cd_W(x)}\int_0^1\int_\R\int_\R W(x,x_\ast)\ave{(pv+qv_\ast)^2}M(x,v)M(x_\ast,v_\ast)\,dv\,dv_\ast\,dx_\ast \\
    &= \ave{p^2}m_2(x)+\frac{\ave{q^2}}{\cd_W(x)}\int_0^1W(x,x_\ast)m_2(x_\ast)\,dx_\ast \\
    &\leq \ave{p^2+q^2}\norm{m_2}{\infty}
\end{align*}
whence $\norm{m_2}{\infty}\leq\ave{p^2+q^2}\norm{m_2}{\infty}$ and the thesis follows.
\end{proof}
Consequently, under the additional assumption $\ave{p^2+q^2}<1$ we can take $\alpha=1$, $\cM_2=0$ and work in $\cP_{1,1,0}(\R)$ with $m_1\equiv 0$ fixed. This space consists solely of the Dirac probability measure $\delta_0(v)$, the equilibrium distribution to which solutions of~\eqref{eq:Boltz-type.W} converge asymptotically in time in the Fourier distance $d_1$ for $\ave{p+q}<\frac{\dWinf}{\norm{\cd_W}{\infty}}$  -- the so-called \textit{cooling} in the jargon of classical kinetic theory. This is the same equilibrium distribution found in the case $s=2$ with $\bar{m}_1^0\equiv 0$.

\subsubsection{The case~\texorpdfstring{$\boldsymbol{s=2}$}{}}
Fixing $\ave{p+q}=1$ and $\bar{m}_1^0\neq 0$ allows us to obtain more interesting long-time asymptotics than just the cooling. Multiplying~\eqref{eq:fixed_point} by $v^2$ and integrating over $v$ by means of~\eqref{eq:Q+.weak} we discover
$$ m_2(x)=\ave{p^2}m_2(x)+2\ave{pq}{(\bar{m}_1^0)}^2+\frac{\ave{q^2}}{\cd_W(x)}\int_0^1W(x,x_\ast)m_2(x_\ast)\,dx_\ast, $$
which, if $\ave{p^2+q^2}<\frac{\dWinf}{\norm{\cd_W}{\infty}}\leq 1$, admits the constant solution $m_2(x)=\bar{m}_2$ with
$$ \bar{m}_2:=\frac{2\ave{pq}}{1-\ave{p^2+q^2}}{(\bar{m}_1^0)}^2. $$
Therefore, this time prospective solutions to~\eqref{eq:fixed_point} may be sought among the distributions with fixed and constant first and second moments prescribed as specified above (in particular, not necessarily equal to $0$). The idea is thus to work in $\cP_{2,\alpha,\cM_{2+\alpha}}(\R)$ setting $m_1\equiv\bar{m}_1^0$ and $m_2\equiv\bar{m}_2$, provided a convenient value of $\alpha>0$ can be identified. To this end, we rely on the following fact:
\begin{lemma} \label{lemma:m3}
If $\ave{p^3+q^3}<1$, for any solution to~\eqref{eq:fixed_point} with $m_1\equiv\bar{m}_1^0$ and $m_2\equiv\bar{m}_2$ the quantity
$$ \cm_3(x):=\int_\R\abs{v}^3M(x,v)\,dv $$
 is uniformly bounded with respect to $M$.
\end{lemma}
\begin{proof}
We multiply both sides of~\eqref{eq:fixed_point} by $\abs{v}^3$ and integrate with respect to $v$, again making use of~\eqref{eq:Q+.weak}:
\begin{align*}
    \cm_3(x) &= \frac{1}{\cd_W(x)}\int_0^1\int_\R\int_\R W(x,x_\ast)\ave{\abs{pv+qv_\ast}^3}M(x,v)M(x_\ast,v_\ast)\,dv\,dv_\ast\,dx_\ast \\
    &\leq \ave{p^3}\cm_3(x)+\frac{1}{\cd_W(x)}\int_0^1W(x,x_\ast)\Bigl(3\bigl(\ave{p^2q}\cm_1(x_\ast)
        +\ave{pq^2}\cm_1(x)\bigr)\bar{m}_2+\ave{q^3}\cm_3(x_\ast)\Bigr)\,dx_\ast.
\end{align*}
The Cauchy--Schwarz inequality implies
$$ \cm_1(x)=\int_\R\abs{v}M(x,v)\,dv\leq{\left(\int_\R v^2M(x,v)\,dv\right)}^{1/2}=\bar{m}_2^{1/2}, $$
therefore the previous estimate specialises further as
\begin{align*}
    \cm_3(x) &\leq \ave{p^3}\cm_3(x)+3\ave{pq(p+q)}\bar{m}_2^{3/2}+\frac{\ave{q^3}}{\cd_W(x)}\int_0^1W(x,x_\ast)\cm_3(x_\ast)\,dx_\ast \\
    &\leq \ave{p^3+q^3}\norm{\cm_3}{\infty}+3\ave{pq(p+q)}\bar{m}_2^{3/2}
\end{align*}
and finally, since $\ave{p^3+q^3}<1$ by assumption,
$$ \norm{\cm_3}{\infty}\leq\frac{3\ave{pq(p+q)}}{1-\ave{p^3+q^3}}\bar{m}_2^{3/2}
    =\frac{3\ave{pq(p+q)}}{1-\ave{p^3+q^3}}{\left(\frac{2\ave{pq}}{1-\ave{p^2+q^2}}\right)}^{3/2}\abs{\bar{m}_1^0}^3. $$
This shows that $\cm_3$ is bounded and that the bound is uniform with respect to any distribution $M$ having the specified first and second moments.
\end{proof}

Owing to Lemma~\ref{lemma:m3}, we can choose $\alpha=1$ and work with distributions in $\cP_{2,1,\cM_3}(\R)$ having fixed first and second moments as set above plus
\begin{equation}
    \cM_3:=\frac{3\ave{pq(p+q)}}{1-\ave{p^3+q^3}}{\left(\frac{2\ave{pq}}{1-\ave{p^2+q^2}}\right)}^{3/2}\abs{\bar{m}_1^0}^3.
    \label{eq:M3}
\end{equation}
Lemma~\ref{lemma:P.complete_metric_space} guarantees that such a space with the distance $d_2$ is complete. We take advantage of such completeness to prove:
\begin{theorem} \label{theo:M.s=2}
Let $\ave{p+q}=1$, $\ave{p^2+q^2}<\frac{\dWinf}{\norm{\cd_W}{\infty}}$ with $\dWinf>0$, and $\ave{p^3+q^3}<1$. Fix $\bar{m}_1^0\in\R$. There exists a unique solution $M(x,\cdot)\in\cP_{2,1,\cM_3}(\R)$ to~\eqref{eq:fixed_point} with mean $\bar{m}_1^0$, the constant $\cM_3$ being fixed as in~\eqref{eq:M3}.
\end{theorem}
\begin{proof}
Owing to the previous discussion, we seek $M$ in the complete metric space $(\cP_{2,1,\cM_3}(\R),\,d_2)$ with
\begin{equation}
    m_1=\bar{m}_1^0, \qquad m_2=\frac{2\ave{pq}}{1-\ave{p^2+q^2}}{(\bar{m}_1^0)}^2, \qquad
        \norm{\cm_3}{\infty}\leq\cM_3.
    \label{eq:m1.m2.m3}
\end{equation}
For this, we apply Banach fixed point theorem to the operator $\frac{1}{\cd_W(x)}Q^+$, cf.~\eqref{eq:fixed_point}.
\begin{enumerate}[label=(\roman*)]
\item First, we check that $\frac{1}{\cd_W(x)}Q^+$ maps $\cP_{2,1,\cM_3}(\R)$ with~\eqref{eq:m1.m2.m3} into itself.

Let $M(x,\cdot)\in\cP_{2,1,\cM_3}(\R)$ fulfil~\eqref{eq:m1.m2.m3}, then:
\begin{align*}
    \int_\R v\frac{1}{\cd_W(x)}Q^+(M,M)(x,v)\,dv &= \ave{p+q}m_1=\bar{m}_1^0, \\
    \int_\R v^2\frac{1}{\cd_W(x)}Q^+(M,M)(x,v)\,dv &= \ave{p^2+q^2}m_2+2\ave{pq}m_1^2 \\
    &= 2\ave{pq}\left(\frac{\ave{p^2+q^2}}{1-\ave{p^2+q^2}}+1\right){(\bar{m}_1^0)}^2=\frac{2\ave{pq}}{1-\ave{p^2+q^2}}{(\bar{m}_1^0)}^2, \\
    \int_\R\abs{v}^3\frac{1}{\cd_W(x)}Q^+(M,M)(x,v)\,dv &\leq \ave{p^3}\cm_3(x)+\frac{3\ave{p^2q}\bar{m}_2}{\cd_W(x)}\int_0^1W(x,x_\ast)\cm_1(x_\ast)\,dx_\ast \\
    &\phantom{\leq} +3\ave{pq^2}\cm_1(x)\bar{m}_2+\frac{\ave{q^3}}{\cd_W(x)}\int_0^1W(x,x_\ast)\cm_3(x_\ast)\,dx_\ast \\
    &\leq \ave{p^3+q^3}\norm{\cm_3}{\infty}+3\ave{pq(p+q)}\norm{\cm_1}{\infty}\bar{m}_2 \\
    &\leq \ave{p^3+q^3}\cM_3+3\ave{pq(p+q)}\bar{m}_2^{3/2} \\
    &= 3\ave{pq(p+q)}\left(\frac{\ave{p^3+q^3}}{1-\ave{p^3+q^3}}-1\right)\bar{m}_2^{3/2}=\cM_3.
\end{align*}

\item Next, we show that, in the stated assumptions, $\frac{1}{\cd_W(x)}Q^+$ is a contraction on $\cP_{2,1,\cM_3}(\R)$.

Let
$$ \frac{1}{\cd_W(x)}\widehat{Q^+}(M,M)(x,\xi):=\frac{1}{\cd_W(x)}\int_\R Q^+(M,M)(x,v)e^{-i\xi v}\,dv $$
the Fourier transform with respect to $v$ of $\frac{1}{\cd_W(x)}Q^+$. If $M_1,\,M_2$ are two distributions in $\cP_{2,1,\cM_3}(\R)$ satisfying~\eqref{eq:m1.m2.m3} then from~\eqref{eq:Q+.weak} with $\varphi(v)=e^{-i\xi v}$ we obtain
\begin{align*}
    & \frac{1}{\cd_W(x)}\widehat{Q^+}(M_2,M_2)(x,\xi)-\frac{1}{\cd_W(x)}\widehat{Q^+}(M_1,M_1)(x,\xi) \\
    &\qquad =\frac{1}{\cd_W(x)}\int_0^1W(x,x_\ast)\left\langle\hat{M}_2(x,p\xi)\left(\hat{M}_2(x_\ast,q\xi)-\hat{M}_1(x_\ast,q\xi)\right)\right. \\
    &\qquad \phantom{=\frac{1}{\cd_W(x)}\int_0^1W(x,x_\ast)} \left.+\left(\hat{M}_2(x,p\xi)-\hat{M}_1(x,p\xi)\right)\hat{M}_1(x_\ast,q\xi)\right\rangle dx_\ast,
\end{align*}
whence
\begin{align*}
    & \frac{\abs*{\frac{1}{\cd_W(x)}\widehat{Q^+}(M_2,M_2)(x,\xi)-\frac{1}{\cd_W(x)}\widehat{Q^+}(M_1,M_1)(x,\xi)}}{\abs{\xi}^2} \\
    &\quad \leq\frac{1}{\cd_W(x)}\int_0^1W(x,x_\ast)\ave*{q^2\frac{\abs{\hat{M}_2(x_\ast,q\xi)-\hat{M}_1(x_\ast,q\xi)}}{\abs{q\xi}^2}
        +p^2\frac{\abs{\hat{M}_2(x,p\xi)-\hat{M}_1(x,p\xi)}}{\abs{p\xi}^2}}dx_\ast \\
    &\quad \leq\frac{\ave{q^2}}{\cd_W(x)}\int_0^1W(x,x_\ast)d_2(M_1(x_\ast),M_2(x_\ast))\,dx_\ast+\ave{p^2}d_2(M_1(x),M_2(x)) \\
    &\quad \leq\ave{p^2+q^2}\sup_{x\in [0,\,1]}d_2(M_1(x),M_2(x))
\end{align*}
and finally
\begin{multline*}
    \sup_{x\in [0,\,1]}d_2\!\left(\frac{1}{\cd_W(x)}Q^+(M_1,M_1)(x),\frac{1}{\cd_W(x)}Q^+(M_2,M_2)(x)\right) \\
        \leq\ave{p^2+q^2}\sup_{x\in [0,\,1]}d_2(M_1(x),M_2(x))
\end{multline*}
with $\ave{p^2+q^2}<\frac{\dWinf}{\norm{\cd_W}{\infty}}\leq 1$. \qedhere
\end{enumerate}
\end{proof}

\subsection{Convergence to graph-free solutions}
Theorem~\ref{theo:M.s=2} asserts the existence and uniqueness of an equilibrium distribution for~\eqref{eq:Boltz.W-strong}, whose mean and energy -- and hence also variance -- are independent of $x$. On the other hand, the theory for the classical homogeneous Boltzmann-type equation with ``all-to-all'', viz. \textit{graph-free}, interactions predicts, under the assumptions $\ave{p+q}=1$, $\ave{p^2+q^2}<1$, $\ave{p^3+q^3}<1$, which encompass those of Theorem~\ref{theo:M.s=2}, an equilibrium distribution with exactly the same mean and energy as those of the distribution $M$ in Theorem~\ref{theo:M.s=2}, see~\cite[Theorem~5.3]{loy2026RMUP}.
\begin{quote}
By  uniqueness, the distribution $M$ of Theorem~\ref{theo:M.s=2} is therefore the equilibrium distribution of the homogeneous Boltzmann-type equation with graph-free interactions.
\end{quote}

This conclusion can be restated by saying that, under the assumptions of Theorem~\ref{theo:M.s=2} together with an $x$-independent initial mean value, the solutions to~\eqref{eq:Boltz.W-strong} converge in time to the same $x$-independent equilibrium distribution as that arising asymptotically in the problem without the graph. Hence, a system of agents undergoing networked interactions may evolve in such a way that the graph of connections becomes negligible in the long run, provided suitable relationships hold between the coefficients $p,\,q$ of the interaction rules~\eqref{eq:int_rules},~\eqref{eq:lin_int} and the graph connectivity parameters $\dWinf,\,\norm{\cd_W}{\infty}$.

To further substantiate these considerations, we now estimate the distance between the solutions of~\eqref{eq:Boltz.W-strong} and those of the corresponding graph-free problem based on the following homogeneous Boltzmann-type equation:
\begin{equation}
    \frac{\partial F}{\partial t}(v,t)=\int_\R\frac{1}{\abs{p^2-q^2}}\ave{F(\pr{v},t)F(\pr{v}_\ast,t)}\,dv_\ast-F(v,t),
    \label{eq:Boltz-strong}
\end{equation}
which in weak form reads
$$ \frac{d}{dt}\int_\R\varphi(v)F(v,t)\,dv=\int_\R\int_\R\ave{\varphi(v')-\varphi(v)}F(v,t)F(v_\ast,t)\,dv\,dv_\ast $$
for every observable $\varphi$, where $F=F(v,t):\R\times (0,\,+\infty)\to\R_+$ is the graph-free, i.e. $x$-independent, kinetic distribution function and $v'$ is given by~\eqref{eq:lin_int}.

Letting $\varphi(v)=e^{-i\xi v}$, we obtain the Fourier-transformed version of this equation:
$$ \frac{\partial\hat{F}}{\partial t}(\xi,t)=\ave{\hat{F}(p\xi,t)\hat{F}(q\xi,t)}-\hat{F}(\xi,t), $$
which is the counterpart of~\eqref{eq:Boltz-type.W.Fourier} in the graph-free setting.

\begin{theorem}
Given $\bar{m}_1^0\in\R$, let $f^0=f^0(x,v)$ and $F^0=F^0(v)$, with $f^0(x,\cdot),\,F^0\in\cP_2(\R)$ for every $x\in [0,\,1]$, be the initial conditions of~\eqref{eq:Boltz.W-strong} and~\eqref{eq:Boltz-strong}, respectively, such that
$$ \int_\R vf^0(x,v)\,dv=\int_\R vF^0(v)\,dv=\bar{m}_1^0, \qquad \forall\,x\in [0,\,1]. $$
Assume moreover that the interaction parameters in~\eqref{eq:lin_int} are such that
$$ \ave{p+q}=1, \qquad \ave{p^2+q^2}<\frac{\dWinf}{\norm{\cd_W}{\infty}}, \qquad \ave{p^3+q^3}<1. $$
Then there exists a constant $C$ such that
\begin{multline}
    \sup_{x\in [0,\,1]}d_2(f(x,t),F(t))\leq\sup_{x\in [0,\,1]}d_2(f^0(x),F^0)e^{\left(\ave{p^2+q^2}\norm{\cd_W}{\infty}-\dWinf\right)t} \\
    +C\norm{1-\cd_W}{\infty}d_2(F^0,M)\left(e^{\left(\ave{p^2+q^2}-1\right)t}-e^{\left(\ave{p^2+q^2}\norm{\cd_W}{\infty}-\dWinf\right)t}\right)
    \label{eq:d2.F-f}
\end{multline}
for all $t>0$, where $M\in\cP_2(\R)$ is the $x$-independent equilibrium distribution of Theorem~\ref{theo:M.s=2}.

In particular,
$$ \lim_{t\to +\infty}\sup_{x\in [0,\,1]}d_2(f(x,t),F(t))=0. $$
\end{theorem}
\begin{proof}
In the stated assumptions, the solutions of both~\eqref{eq:Boltz.W-strong} and~\eqref{eq:Boltz-strong} have constant mean $\bar{m}_1^0$, therefore, owing to Lemma~\ref{lemma:ds.finite}, the distance $d_2$ between them is well defined.

Subtracting the Fourier-transformed versions of these equations yields
\begin{align*}
    \frac{\partial}{\partial t}\left(\hat{F}(\xi,t)-\hat{f}(x,\xi,t)\right) &=
        \ave{\hat{F}(p\xi,t)\hat{F}(q\xi,t)}-\hat{F}(\xi,t) \\
    &\phantom{=} -\int_0^1W(x,x_\ast)\ave{\hat{f}(x,p\xi,t)\hat{f}(x_\ast,q\xi,t)}\,dx_\ast+\cd_W(x)\hat{f}(x,\xi,t).
\intertext{To bring the terms on the right-hand side into closer correspondence, we write $\ave{\hat{F}(p\xi,t)\hat{F}(q\xi,t)}-\hat{F}(\xi,t)=(\ave{\hat{F}(p\xi,t)\hat{F}(q\xi,t)}-\hat{F}(\xi,t))(1-\cd_W(x)+\cd_W(x))$, whence}
    &= (1-\cd_W(x))\left(\ave{\hat{F}(p\xi,t)\hat{F}(q\xi,t)}-\hat{F}(\xi,t)\right) \\
    &\phantom{=} +\int_0^1W(x,x_\ast)\ave{\hat{F}(p\xi,t)\hat{F}(q\xi,t)-\hat{f}(x,p\xi,t)\hat{f}(x_\ast,q\xi,t)}\,dx_\ast \\
    &\phantom{=} -\cd_W(x)\left(\hat{F}(\xi,t)-\hat{f}(x,\xi,t)\right),
\intertext{where we have used that $\cd_W(x)=\int_0^1W(x,x_\ast)\,dx_\ast$. Moreover, considering that the equilibrium distribution $M$ satisfies $\ave{\hat{M}(p\xi)\hat{M}(q\xi)}-\hat{M}(\xi)=0$ (to see this, use the Fourier-transformed version of~\eqref{eq:fixed_point} together with the fact that $M$ is independent of $x$), we may further manipulate the right-hand side and rewrite it as}
    &= (1-\cd_W(x))\left[\ave{\hat{F}(p\xi,t)\hat{F}(q\xi,t)-\hat{M}(p\xi)\hat{M}(q\xi)}-\left(\hat{F}(\xi,t)-\hat{M}(\xi)\right)\right] \\
    &\phantom{=} +\int_0^1W(x,x_\ast)\ave{\hat{F}(p\xi,t)\hat{F}(q\xi,t)-\hat{f}(x,p\xi,t)\hat{f}(x_\ast,q\xi,t)}\,dx_\ast \\
    &\phantom{=} -\cd_W(x)\left(\hat{F}(\xi,t)-\hat{f}(x,\xi,t)\right).
\end{align*}

Setting
$$ h(x,\xi,t):=\frac{\hat{F}(\xi,t)-\hat{f}(x,\xi,t)}{\abs{\xi}^2}, \qquad
    H(\xi,t):=\frac{\hat{F}(\xi,t)-\hat{M}(\xi)}{\abs{\xi}^2} $$
for ease of notation and proceeding like in some of the previous proofs, we deduce
\begin{align*}
    \frac{\partial h}{\partial t}(x,\xi,t) &+ \cd_W(x)h(x,\xi,t) \\
    &= (1-\cd_W(x))\left[\ave*{q^2\hat{F}(p\xi,t)H(q\xi,t)+p^2H(p\xi,t)\hat{M}(q\xi)}-H(\xi,t)\right] \\
    &\phantom{=} +\int_0^1W(x,x_\ast)\ave*{q^2\hat{F}(p\xi,t)h(x_\ast,q\xi,t)+p^2h(x,p\xi,t)\hat{f}(x_\ast,q\xi,t)}\,dx_\ast,
\end{align*}
whence, multiplying both sides by $e^{\cd_W(x)t}$ and integrating over $[0,\,t]$, with $t>0$,
\begin{align*}
    e^{\cd_W(x)t}\abs{h(x,\xi,t)} &\leq \abs{h(x,\xi,0)} \\
    &\phantom{\leq} +\norm{1-\cd_W}{\infty}\int_0^te^{\cd_W(x)\tau}\left[\ave*{q^2\abs{H(q\xi,\tau)}+p^2\abs{H(p\xi,t)}}+\abs{H(\xi,\tau)}\right]d\tau \\
    &\phantom{\leq} +\int_0^te^{\cd_W(x)\tau}\int_0^1W(x,x_\ast)\ave*{q^2\abs{h(x_\ast,q\xi,\tau)}+p^2\abs{h(x,p\xi,\tau)}}\,dx_\ast\,d\tau.
\intertext{Noting that $\sup_{\xi\in\R\setminus\{0\}}\abs{h(x,\xi,t)}=d_2(f(x,t),F(t))$ and $\sup_{\xi\in\R\setminus\{0\}}\abs{H(\xi,t)}=d_2(F(t),M)$, we continue this chain of inequalities as}
    &\leq \sup_{x\in [0,\,1]}d_2(f^0(x),F^0) \\
    &\phantom{\leq} +\norm{1-\cd_W}{\infty}\ave{p^2+q^2+1}\int_0^te^{\cd_W(x)\tau}d_2(F(\tau),M)\,d\tau \\
    &\phantom{\leq} +\ave{p^2+q^2}\norm{\cd_W}{\infty}\int_0^te^{\cd_W(x)\tau}\sup_{x\in [0,\,1]}d_2(f(x,\tau),F(\tau))\,d\tau
\end{align*}
and finally, by the same argument on $e^{\cd_W(x)t}$ as that used in the proof of Lemma~\ref{lemma:m1N},
\begin{align*}
    e^{\dWinf t}\sup_{x\in [0,\,1]}d_2(f(x,t),F(t)) &\leq \sup_{x\in [0,\,1]}d_2(f^0(x),F^0) \\
    &\phantom{\leq} +\norm{1-\cd_W}{\infty}\ave{p^2+q^2+1}\int_0^te^{\dWinf\tau}d_2(F(\tau),M)\,d\tau \\
    &\phantom{\leq} +\ave{p^2+q^2}\norm{\cd_W}{\infty}\int_0^te^{\dWinf\tau}\sup_{x\in [0,\,1]}d_2(f(x,\tau),F(\tau))\,d\tau.
\end{align*}

At this point, Gr\"{o}nwall's inequality in the form given in~\cite[Lemma~6.2]{loy2026RMUP} applied to the function $t\mapsto e^{\dWinf t}\sup_{x\in [0,\,1]}d_2(f(x,t),F(t))$ yields
\begin{align*}
    \sup_{x\in [0,\,1]}d_2(f(x,t),F(t)) &\leq \sup_{x\in [0,\,1]}d_2(f^0(x),F^0)e^{\left(\ave{p^2+q^2}\norm{\cd_W}{\infty}-\dWinf\right)t} \\
    &\phantom{\leq} +\norm{1-\cd_W}{\infty}\ave{p^2+q^2+1}\int_0^te^{\left(\ave{p^2+q^2}\norm{\cd_W}{\infty}-\dWinf\right)(t-\tau)}
        d_2(F(\tau),M)\,d\tau.
\end{align*}
Moreover, from the theory for the graph-free equation~\eqref{eq:Boltz-strong} we know that
$$ d_2(F(t),M)\leq d_2(F^0,M)e^{\left(\ave{p^2+q^2}-1\right)t}, $$
cf.~\cite[Proposition~5.1]{loy2026RMUP}. Substituting this into the previous inequality and evaluating the integral on the right-hand side explicitly gives~\eqref{eq:d2.F-f} with
$$ C=\frac{\ave{p^2+q^2+1}}{\ave{p^2+q^2}(1-\norm{\cd_W}{\infty})+\dWinf-1}. $$
The convergence of $f(x,\cdot,t)$ to $F(\cdot,t)$ as $t\to +\infty$ follows then from~\eqref{eq:d2.F-f}, owing to the fact that $\ave{p^2+q^2}<\frac{\dWinf}{\norm{\cd_W}{\infty}}\leq 1$ implies simultaneously $\ave{p^2+q^2}\norm{\cd_W}{\infty}-\dWinf<0$ and $\ave{p^2+q^2}-1<0$.
\end{proof}

We observe that if $F^0=M$ then~\eqref{eq:d2.F-f} coincides with the estimate provided by Theorem~\ref{theo:ds} with $g\equiv M$, recalling that $M$ is also the unique equilibrium solution of~\eqref{eq:Boltz.W-strong}. Conversely, if $F^0\neq M$ the distance between $f$ and $F$ is estimated by two parts: on the one hand, the distance between the initial conditions $f^0,\,F^0$; on the other hand, the distance of the graphon $W$ from the \textit{complete graphon}, i.e. the function identically equal to $1$ on $[0,\,1]^2$. This latter distance is captured by the term $\norm{1-\cd_W}{\infty}$, since the degree function is identically equal to $1$ on the interval $[0,\,1]$ if and only if $W$ is the complete graphon on the unit square.

In particular, we note that if $W$ is the complete graphon, then agents' interactions fall within the ``all-to-all'' case and the collisional operator in~\eqref{eq:Boltz.W-strong} can be marginalised with repect to $x_\ast$. Yet, if the initial condition $f^0$ depends on $x$ then the solution of~\eqref{eq:Boltz.W-strong} does not coincide, at every time, with that of~\eqref{eq:Boltz-strong}, although it decays to the latter exponentially fast in time.

If, instead, one prescribes the same initial condition to both~\eqref{eq:Boltz.W-strong} and~\eqref{eq:Boltz-strong}, choosing in particular $f^0$ independent of $x$, but $W$ is not the complete graphon, then the corresponding solutions still do not coincide at every time due to the term proportional to $\norm{1-\cd_W}{\infty}$ in~\eqref{eq:d2.F-f}, which however decays in turn exponentially fast in time.

Only when $f^0\equiv F^0$ and $W$ is the complete graphon one gets from~\eqref{eq:d2.F-f} that $f\equiv F$, thus in particular that $f$ is $x$-independent, at every time.


\section*{Acknowledgements}
AT is member of GNFM (Gruppo Nazionale per la Fisica Matematica) of INdAM (Istituto Nazionale di Alta Matematica), Italy.
	
\bibliographystyle{plain}
\bibliography{biblio}
\end{document}